\documentclass[a4paper,11pt]{article}

\usepackage{fullpage}
\usepackage{times}
\usepackage{soul}
\usepackage{url}
\usepackage[utf8]{inputenc}
\usepackage[small]{caption}
\usepackage{graphicx}
\usepackage[dvipsnames]{xcolor}
\usepackage{amsmath,amsthm,amssymb,amsfonts}
\usepackage{booktabs}
\usepackage{algorithm}
\usepackage{dsfont}
\usepackage{tikz}
\usepackage{tcolorbox}
\usepackage[T1]{fontenc}
\usepackage{enumitem}

\usepackage{algorithm}
\usepackage{algorithmic}
\usepackage{bbm}

\usepackage{thmtools}
\usepackage{thm-restate}

\usepackage{cleveref}

\newtheorem{theorem}{Theorem}[section]

\newtheorem{proposition}[theorem]{Proposition}

\theoremstyle{definition}
\newtheorem{definition}[theorem]{Definition}
\newtheorem{example}[theorem]{Example}
\newtheorem{remark}[theorem]{Remark}

\usepackage{natbib}

\tcbset{size=small,top=1mm,bottom=1mm,middle=1mm}

\usepackage{booktabs}

\usepackage{tikz}
\usetikzlibrary{matrix}
\usetikzlibrary{patterns}
\makeatletter
\tikzset{
        hatch distance/.store in=\hatchdistance,
        hatch distance=5pt,
        hatch thickness/.store in=\hatchthickness,
        hatch thickness=5pt
        }
\pgfdeclarepatternformonly[\hatchdistance,\hatchthickness]{north east hatch}
    {\pgfqpoint{-1pt}{-1pt}}
    {\pgfqpoint{\hatchdistance}{\hatchdistance}}
    {\pgfpoint{\hatchdistance-1pt}{\hatchdistance-1pt}}%
    {
        \pgfsetcolor{\tikz@pattern@color}
        \pgfsetlinewidth{\hatchthickness}
        \pgfpathmoveto{\pgfqpoint{0pt}{0pt}}
        \pgfpathlineto{\pgfqpoint{\hatchdistance}{\hatchdistance}}
        \pgfusepath{stroke}
    }
\makeatother

\usepackage{mathtools}

\newcommand{\sat}{\mathrm{sat}}

\newcommand{\apprvls}{\mathbf{a}}
\newcommand{\outcome}{\mathbf{o}}
\newcommand{\calE}{\mathcal{E}}
\newcommand{\calF}{\mathcal{F}}
\newcommand{\calA}{\mathcal{A}}
\newcommand{\calC}{\mathcal{C}}
\newcommand{\calN}{\mathcal{N}}
\newcommand{\calI}{\mathcal{I}}

\DeclarePairedDelimiter{\floor}{\lfloor}{\rfloor}
\DeclarePairedDelimiter{\ceil}{\lceil}{\rceil}

\allowdisplaybreaks

\usepackage{authblk}

\title{\bf Strengthening Proportionality in Temporal Voting}

\author[1]{Bradley Phillips}
\author[2]{Edith Elkind}
\author[1]{Nicholas Teh}
\author[1]{Tomasz Wąs}
\affil[1]{University of Oxford, UK}
\affil[2]{Northwestern University, USA}

\date{\vspace{-10mm}}
\begin{document}
\maketitle

\begin{abstract}
    We study proportional representation in the framework of temporal voting with approval ballots. Prior work adapted basic proportional representation concepts---justified representation (JR), proportional JR (PJR), and extended JR (EJR)---from the multiwinner setting to the temporal setting. Our work introduces and examines ways of going beyond EJR. Specifically, we consider stronger variants of JR, PJR, and EJR, and introduce temporal adaptations of more demanding multiwinner axioms, such as EJR+, full JR (FJR), full proportional JR (FPJR), and the Core. For each of these concepts, we investigate its existence and study its relationship to existing notions, thereby establishing a rich hierarchy of proportionality concepts. Notably, we show that two of our proposed axioms---EJR+ and FJR---strengthen EJR while remaining satisfiable in every temporal election.

\end{abstract}

\section{Introduction}
\emph{Proportional representation} is a fundamental principle in the design of fair decision-making mechanisms, particularly in the framework of \emph{multiwinner voting}, where the goal is to select a representative group of candidates based on the preferences of the electorate.
For approval ballots, this principle is instantiated as the axiom of \emph{justified representation (JR)} and its extensions, such as PJR/EJR/FJR/EJR+ and the core \citep{aziz2017jr,brill2023robust,peters2021fjr,sf2017pjr}. 
These notions are motivated by classical social choice contexts, from electing diverse panels and boards to ensuring the presence of minority voices in political bodies \citep{lackner2022abc,phragmen}.
More recently, they have been shown to be relevant for machine learning (ML) systems that incorporate human feedback, optimize group-level utility, or seek fairness across populations. Examples include increasing diversity in recommendation systems and social media~\citep{gawron2024movies,revel2025fb,streviniotis2022tourismrecsys}, ensuring balanced representation in committee-based decisions in blockchain governance \citep{boehmer2024polkadot,polkadot2020,cevallos2021polkadot}, and offering representation guarantees in LLM-augmented democratic processes \citep{boehmer2025gensocialchoicenext,fish2024generativesocialchoice}.

However, beyond static selection tasks, many ML applications involve decision-making over multiple rounds, with preferences that may change over time.
For instance, curriculum learning constructs sequences of training tasks for reinforcement learning agents; temporal fairness guarantees can prevent the curriculum from overemphasizing early-stage objectives at the expense of later goals \citep{wu2024marlcurriculum, yang2021marlcurriculum}.
Similarly, streaming service catalogs and content recommendation systems periodically refresh their offerings: ensuring proportionality across time ensures that under-served genres or user communities are not systematically sidelined across updates \citep{chen2024multiwinnerreconfiguration}.
In the domain of generative AI, blending outputs from multiple models---each with its own inductive biases---calls for proportional merging strategies that preserve diversity of generated content over time \citep{peters2024proportionalECAI}.

In these applications, proportional representation needs to be maintained not just in a single round, but across the entire time horizon. This calls for an adaptation of the JR axioms to the temporal setting.

Recently, \citet{bulteau2021jrperpetual}, 
\citet{chandak2024proportional}, and \citet{elkind2025verifying}
considered temporal voting with approval ballots, where the goal is to select one candidate per round. 
They adapted some of the more established justified representation axioms (namely, JR, PJR and EJR) to this setting, formulated temporal variants of popular multiwinner rules, e.g., Proportional Approval Voting \citep{kilgour2010approval}, Method of Equal Shares \citep{peters2020proportionality}, or Greedy Cohesive Rule \citep{bredereckF0N19}, and checked whether their outputs satisfy temporal JR axioms, as well as considered the complexity of verifying whether an outcome satisfies a given axiom.
However, they did not attempt to extend more demanding proportionality axioms such as EJR+, FJR, and core stability to the temporal setting. 

\subsection{Our Contributions}
In this work, we aim to identify the most ambitious justified representation-style axioms that can be satisfied in temporal voting, and, more broadly, 
to create a comprehensive map of the landscape of temporal JR axioms. 

In \Cref{sec:ejr-plus}, we focus on the EJR+ axiom~\citep{brill2023robust}. We propose an adaptation that preserves the three main attractive features of its multiwinner variant: strengthening EJR, being verifiable in polynomial time, and being satisfiable by a polynomial-time computable voting rule. Our analysis in this section highlights the difficulties of extending proportionality axioms from the multiwinner setting to the temporal setting: the most straightforward variant of temporal EJR+ turns out to be unsatisfiable in general, and we explain why this is the case by uncovering a hidden layer of our axiomatic hierarchy.

In \Cref{sec:fjr}, we pursue the same approach for the FJR axiom~\citep{peters2021fjr}, and show that an outcome satisfying temporal FJR always exists and can be found by a modification of the Greedy Cohesive Rule.

In \Cref{sec:fpjr}, we look at the recently proposed notion of full proportional justified representation (FPJR) \citep{kalayci2025fpjr} that combines the FJR notion of group cohesion with a PJR-style collective guarantee.
We adapt it to our setting and show that if the number of voters divides the number of rounds and each voter approves at least one candidate per round, an outcome that satisfies a strong version of FPJR (sFPJR) can be found in polynomial time by the Serial Dictatorship Rule.
Importantly, EJR+ and FJR are stronger than EJR, and sFPJR is incomparable with EJR and the other two axioms. Thus, each of our results advances the frontier of satisfiable proportionality guarantees in the temporal setting.

Finally, in \Cref{sec:core}, we complete the picture of proportionality notions in the temporal setting by formulating temporal core stability axioms.

For all axioms that we define, we establish that they satisfy the expected implications (e.g., temporal FJR implies temporal EJR, just like in the multiwinner setting); for cases where implications do not hold, we provide concrete counterexamples. 
Thus, our work establishes a rich hierarchy of proportionality concepts in temporal voting, illustrated in \Cref{fig:axiom_map}.
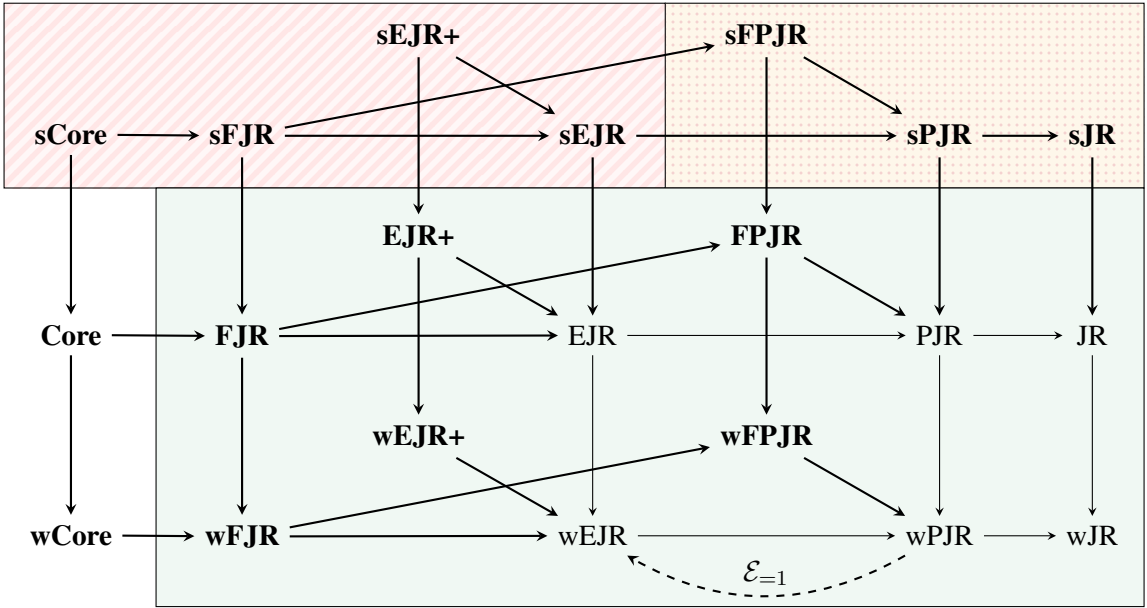
\begin{figure*}[t]
\begin{center}
\begin{tikzpicture}
[Pattern/.style={pattern=north east hatch, pattern color=red!10, hatch distance=7pt, hatch thickness=2pt}]
    \draw [draw=black, preaction={fill=Dandelion!10}, pattern = dots, pattern color=BrickRed!15] (1.2,1.3) rectangle (7.5,3.75);
    \draw [draw=black, preaction={fill=red!03}, Pattern] (-7.5,1.3) rectangle (1.2,3.75);
    \draw [draw=black, fill=Green!05] (-5.5,-4.25) rectangle (7.5,1.3);
  \matrix(m)[matrix of math nodes,row sep=2em,column sep=2.5em,minimum width=2em]
  { & & \text{{\bf sEJR+}} & & \text{{\bf sFPJR}} & &\\
    \text{{\bf sCore}} & \text{{\bf sFJR}} & & \text{{\bf sEJR}} & & \text{{\bf sPJR}} & \text{{\bf sJR}} \\
    & & \text{{\bf EJR+}} & & \text{{\bf FPJR}} & & \\
    \text{{\bf Core}} & \text{{\bf FJR}} & & \text{{EJR}} & & \text{{PJR}} & \text{{JR}} \\
    & & \text{{\bf wEJR+}}& & \text{{\bf wFPJR}} & &\\
    \text{{\bf wCore}} & \text{{\bf wFJR}} & & \text{{wEJR}} & & \text{{wPJR}} & \text{{wJR}} \\};
  \path[-stealth, thick]
    (m-1-3) edge (m-2-4) edge (m-3-3)
    (m-1-5) edge (m-2-6) edge (m-3-5) 
    (m-2-1) edge (m-2-2) edge (m-4-1)
    (m-2-2) edge (m-2-4) edge (m-4-2) edge (m-1-5)
    (m-2-4) edge (m-2-6) edge (m-4-4)
    (m-2-6) edge (m-2-7) edge (m-4-6)
    (m-2-7) edge (m-4-7)
    (m-3-3) edge (m-4-4) edge (m-5-3)
    (m-3-5) edge (m-4-6) edge (m-5-5)
    (m-4-1) edge (m-4-2) edge (m-6-1)
    (m-4-2) edge (m-4-4) edge (m-6-2) edge (m-3-5)
    (m-5-3) edge (m-6-4)
    (m-5-5) edge (m-6-6)
    (m-6-1) edge (m-6-2)
    (m-6-2) edge (m-6-4) edge (m-5-5)
    (m-6-6) edge[dashed,out=-150,in=-30] node [above] {$\calE_{=1}$} (m-6-4)
    ;
  \path[-stealth]
    (m-4-4) edge (m-4-6) edge (m-6-4)
    (m-4-6) edge (m-4-7) edge (m-6-6)
    (m-4-7) edge (m-6-7)
    (m-6-4) edge (m-6-6)
    (m-6-6) edge (m-6-7)
    ;
\end{tikzpicture}
\caption{Axioms considered in our paper.
A solid arrow from axiom $A$ to axiom $B$
means that $A$ implies $B$;
an absence of a path from $A$ to $B$
means that the implication does not hold
(even if each voter approves exactly one candidate per round, denoted by $\calE_{=1}$).
A dashed arrow means implication for $\calE_{=1}$,
but not in general.
Thick arrows denote implications established in this paper.
The axioms on a green plain background
can be satisfied in each election.
The axioms on a yellow dotted background
can be satisfied if each voter approves at least one candidate per round and the number of rounds is divisible by the number of voters,
but not in general;
for the ones on a red striped background
even in this special case
there are instances where they are not satisfiable.
For axioms on white background, 
satisfiability remains open.
The axioms introduced in this paper
are written in bold.}
\label{fig:axiom_map}
\end{center}
\vspace{-0.15in}
\end{figure*}
All missing proofs are included in the technical appendix.

\subsection{Related Work}
Proportional representation in temporal voting was first studied by \citet{bulteau2021jrperpetual}, who extended two notions of proportional representation---JR and PJR---to the temporal setting. 
They showed that a JR outcome can be computed in polynomial time, even for the most demanding version of JR among those they considered. 
For PJR, they proved the existence of an outcome satisfying the axiom, but their constructive proof relies on an exponential-time algorithm.
This work was extended by \citet{chandak2024proportional}, who introduced a temporal variant of EJR and showed that several multiwinner voting rules that satisfy PJR or EJR in the standard model can be adapted to the temporal setting so as to satisfy the corresponding temporal variants of PJR and EJR.
Subsequently, \citet{elkind2025verifying} studied the complexity of verifying whether an outcome satisfies various proportionality axioms.

\citet{elkind2024temporalelections} focused on the tradeoffs between the social welfare and other desirable properties in the temporal setting, such as strategyproofness and a weaker form of proportionality. \citet{elkind2025temporalchores} considered similar computational questions, but in the setting where voters have disutility over candidates.
A similar model has also been explored by \citet{lackner2020perpetual} and \citet{lackner2023proportionalPV} under the \emph{perpetual voting} framework, where they consider temporal extensions of  multiwinner voting rules and their axiomatic properties.
We refer the reader to \citet{elkind2024temporalsurvey} for a systematic review of recent work on temporal voting and related models.
An adjacent line of work looks into sequential committee elections where an entire committee is elected at each round \citep{bredereck2020successivecommittee,bredereck2022committeechange,chen2024multiwinnerreconfiguration,deltl2023seqcommittee,zech2024multiwinnerchange}.
These works focus on constraining the extent of changes to the committees chosen across rounds.

Multiple other subareas are similar in spirit to the temporal voting framework, but with additional constraints/restrictions. In \emph{apportionment with approval preferences}, the goal is to allocate the seats of a fixed-size committee to parties based on voters' (approval) preferences over the parties \citep{brill2024partyapportionment,delemazure2022spelection}. This is equivalent to a restricted setting of temporal voting with static voter preferences.
In \emph{fair scheduling}, each agent's preference is a permutation, and so is the outcome \citep{elkind2022temporalslot,patro2022virtualconf}.
Our model differs from this setting in that we allow each project to be chosen more than once (both in agents' preferences and the outcome). \emph{Fair public decision-making} is similar to the temporal voting model, but the focus is typically on weaker forms of proportionality \citep{alouf2022better,conitzer2017fairpublic,fain2018publicgoods,skowron2022proppublic}.
Finally, an adjacent but related framework is that of \emph{resource allocation over time} \citep{allouah2023fairallocationtime,bampis2018fairtime,elkind2025temporalfairdivision}, where items are sequentially (and privately) allocated to agents. While this setting can potentially be modeled using the temporal voting framework, the fairness notions typically considered are fundamentally different.

Another relevant work is that of \citet{MasPieSko-2023-GeneralProportionality}, who considered a more general voting model incorporating feasibility constraints, and proposed proportionality axioms for this model. We include a detailed discussion of the connection between their concepts and some of ours in \Cref{app:general-prop}.

\section{Preliminaries} \label{sec:preliminaries}

For every natural number $k \in \mathbb{N}$, we let $[k] = \{1, 2, \dots, k\}$.
A \emph{temporal election} is a tuple
$E = (C, N, \ell, A)$, where
$C$ is the set of {\em candidates},
$N$ is the set of $n$ \emph{voters},
$\ell$ is the number of \emph{rounds},
and $A = (\apprvls_i)_{i \in N}$ is an \emph{approval profile},
where $\apprvls_i=(a_{i,1}, a_{i,2},\dots, a_{i,\ell})$ is the \emph{ballot} of voter $i \in N$, i.e., 
a list of $\ell$ subsets of $C$; 
in round $r \in [\ell]$ voter $i$
\emph{approves} candidates in $a_{i, r} \subseteq C$.
We denote the set of all temporal elections by $\calE$
and the set of temporal elections in which in every round
every voter approves at least (resp. exactly) one candidate by $\calE_{\ge 1}$ (resp. $\calE_{=1}$).
We say that voters from a subset $S \subseteq N$ \emph{agree} in a round $r \in [\ell]$ if there exists a candidate that they all approve in this round, i.e., $\bigcap_{i \in S} a_{i,r} \neq \varnothing$.

An \emph{outcome}
\(
    \outcome = (o_1, o_2, \dots, o_\ell) \in
    C^\ell
\)
is a sequence of candidates chosen in each round.
For a subset of rounds $R \subseteq [\ell]$ and an outcome $\outcome$, we write $\outcome_R = (o_r)_{r \in R}$ to denote the \emph{suboutcome} with respect to $R$.
The \emph{satisfaction} of a subset of voters $S \subseteq N$ from a suboutcome $\outcome_R$ is $\sat_S(\outcome_R) = |\{r \in R : o_r \in \bigcup_{i\in S} a_{i,r}\}|$, i.e., the number of rounds in $R$ in which the selected candidate is approved by at least one voter from $S$.
If $S = \{i\}$ for some $i \in N$,
we drop the brackets and write
$\sat_i(\outcome_R)$.
A \emph{voting rule} $f$ is a function that for every election $E$ outputs a non-empty set of outcomes $f(E)$.

\subsection*{Proportionality Axioms from Prior Work}
The temporal voting setting can be seen as an extension of {\em multiwinner voting}, where, in a single round of voting, voters in $N$ report approvals $(a_i)_{i\in N}$ over candidates in $C$, and the goal is to select a set of $k$ winners, $W$. The satisfaction of a group of voters $S$ is then defined as the number of candidates in $W$ approved by at least one member of $S$. Key proportionality axioms for multiwinner voting \citep{aziz2017jr,sf2017pjr} require $W$ to provide a certain level of satisfaction to each group of voters $S$ that agree on $t\ge 1$ candidates (in the sense that $|\bigcap_{i\in S}a_i|\ge t$). In particular, 
{\em justified representation (JR$^\textit{\,mw}$)} requires the satisfaction to be at least $1$ if $k\cdot |S|/n\ge 1$, {\em proportional justified representation (PJR$^\textit{\,mw}$)} 
strengthens this guarantee to $\min(t,\floor{k \cdot |S|/n})$, while \emph{extended justified representation} (EJR$^\textit{\,mw}$) requires at least one voter in $S$
to approve $\min(t,\floor{k \cdot |S|/n})$ members of $W$.

\citet{bulteau2021jrperpetual} and \citet{chandak2024proportional} propose two versions of each of these axioms for the temporal setting.
We will now present these axioms using the terminology of \citet{elkind2025verifying}.\footnote{In their conference paper, \citet{chandak2024proportional} use `JR/PJR/EJR' for what we call `weak JR/PJR/EJR' and `strong JR/PJR/EJR' for what we call `JR/PJR/EJR', but their arXiv version uses the same terminology as we do.}
Note that the temporal setting is more restrictive as compared to the multiwinner setting, as we cannot select two candidates in the same round.

The weakest version of these axioms  only considers subsets of voters that agree in every round.
\begin{definition}
\label{def:wJR:wPJR:wEJR:combined}
    Given a temporal election $E =(C, N, \ell, A )$ and an outcome $\outcome$, if for every subset of voters $S \subseteq N$ that agree in all $\ell$ rounds it holds that
    \begin{itemize}[leftmargin=0.3in]
        \item $\sat_{S}(\outcome) \geq \min(1, \floor{\ell \cdot |S|/n})$,
        then $\outcome$ provides \emph{weak justified representation} (wJR),
        \item $\sat_{S}(\outcome) \geq \floor{\ell \cdot |S|/n}$,
        then $\outcome$ provides \emph{weak proportional justified representation} (wPJR),
        \item $\sat_{i}(\outcome) \geq \floor{\ell \cdot |S|/n}$ for some $i\in S$,
        then $\outcome$ provides \emph{weak extended justified representation} (wEJR).
    \end{itemize}
\end{definition}

Definition~\ref{def:wJR:wPJR:wEJR:combined} offers no guarantees to groups that disagree even in a single round. A more flexible approach is to provide guarantees that scale with the number of rounds where the members agree.

\begin{definition}
\label{def:JR:PJR:EJR:combined}
    Given a temporal election
     $E =( C ,N , \ell, A )$
    and an outcome $\outcome$, if for each $t>0$ and every subset of voters $S \subseteq N$
    that agree in a size-$t$ subset of rounds it holds that
    \begin{itemize}[leftmargin=0.3in]
        \item $\sat_{S}(\outcome) \geq \min(1, \floor{t \cdot |S|/n})$,
        then $\outcome$ provides \emph{justified representation} (JR),
        \item $\sat_{S}(\outcome) \geq \floor{t \cdot |S|/n}$,
        then $\outcome$ provides \emph{proportional justified representation} (PJR),
        \item $\sat_{i}(\outcome) \geq \floor{t \cdot |S|/n}$ for some $i\in S$,
        then $\outcome$ provides \emph{extended justified representation} (EJR).
    \end{itemize}
\end{definition}

\citet{chandak2024proportional} show that every temporal election admits an outcome that provides EJR. Also, EJR implies PJR, PJR implies JR, and each axiom implies its weak counterpart (we say that an axiom $A_1$ implies axiom $A_2$,
denoted $A_1 \implies A_2$, if every outcome that provides $A_1$ also provides $A_2$).

\section[Extended Justified Representation +]{Extended Justified Representation\,+}
\label{sec:ejr-plus}
In the context of multiwinner voting, 
\citet{brill2023robust} have recently strengthened EJR$^\textit{\,mw}$ to a new axiom, which they call
\emph{Extended Justified Representation\,+ (EJR$^\textit{\,mw}$+)}. 
The goal of this section is to adapt this axiom to the temporal setting.
EJR$^\textit{\,mw}$+ has three attractive features, which we aim to preserve in our temporal adaptation:
(i) it implies EJR$^\textit{\,mw}$ (while the converse is not true)
(ii) a committee that provides EJR$^\textit{\,mw}$+ can be found in polynomial time, and
(iii) it can be verified in polynomial time whether a given committee provides EJR$^\textit{\,mw}$+.

EJR$^\textit{\,mw}$+ requires that for every group of voters $S$ that jointly approve a candidate $c$, either $c$ is included in the winning committee $W$ or
there is a voter $i \in S$ whose satisfaction from $W$ is at least $\floor{k\cdot |S|/n}$. A direct translation of this axiom to the temporal setting gives us the following definition.

\begin{definition}\label{def:sEJR+}
    Given a temporal election
    $E =( C, N, \ell, A )$, 
    an outcome $\outcome$ provides \emph{strong extended justified representation\,+} (sEJR+)
    if for every subset of voters $S \subseteq N$ and every  round $r \in [\ell]$ with
    $\bigcap_{i \in S} a_{i,r}\neq\varnothing$
    it holds that
    \begin{equation*}
    \text{(i) } \ \sat_i(\outcome) \ge \floor{\ell \cdot |S|/n} \text{ for some } i \in S, \quad \text{or} \quad \text{(ii) } \ o_r \in \textstyle{\bigcap_{i \in S} a_{i,r}}.
\end{equation*}
\end{definition}
Note that the second condition is phrased in terms of choosing an outcome from $\bigcap_{i\in S}a_{i, r}$ rather than selecting a specific candidate;
this is because there may be several candidates approved by all members of $S$ in round $r$, but only one of them can be selected. In contrast, in the multiwinner setting, if $S$ agrees on multiple candidates, several such candidates can be included in the winning committee.

We labeled the concept introduced in \Cref{def:sEJR+} as `strong' for two reasons. First, there are elections (even in $\calE_{=1}$) with no sEJR+ outcomes. Second, 
sEJR+ is fundamentally different from JR/PJR/EJR, as defined in \Cref{def:JR:PJR:EJR:combined}. To show this, we will now define strong JR/PJR/EJR following the approach of \Cref{def:sEJR+}, and explain how the resulting notions differ from JR/PJR/EJR.

\begin{definition}
\label{def:sJR:sPJR:sEJR:combined}
    Given a temporal election
     $E =(C, N, \ell, A )$
    and an outcome $\outcome$,
    if for each $t>0$ and every subset of voters $S \subseteq N$
    that agree in a size-$t$ subset of rounds it holds that
    \begin{itemize}[leftmargin=0.3in]
        \item $\sat_{S}(\outcome) \geq \min(1, \floor{\ell \cdot |S|/n})$,
        then $\outcome$ provides \emph{strong justified representation} (sJR),
        \item $\sat_{S}(\outcome) \geq \min(t, \floor{\ell \cdot |S|/n})$,
        then $\outcome$ provides \emph{strong proportional justified representation} (sPJR),
        \item $\sat_{i}(\outcome) \geq \min(t, \floor{\ell \cdot |S|/n})$ for some $i \in S$,
        then $\outcome$ provides \emph{strong extended justified representation} (sEJR).
    \end{itemize}
\end{definition}

Together with sEJR+, strong JR/PJR/EJR form a natural hierarchy and 
imply their counterparts from \Cref{def:JR:PJR:EJR:combined}.

\begin{restatable}{proposition}{PropVsImplications}\label{prop:vs:implications}
    It holds that: (i) sEJR+ $\implies$ sEJR, (ii) sEJR $\implies$ sPJR and EJR, (iii) sPJR $\implies$ sJR and PJR, and (iv) sJR $\implies$ JR.
\end{restatable}

However, there exist elections (even in $\calE_{=1}$) for which
even the weakest of these axioms, i.e., sJR, is impossible to satisfy. By \Cref{prop:vs:implications},
this impossibility extends to sPJR/sEJR/sEJR+.

\begin{restatable}{proposition}{sjrnotsatisfiable}
\label{prop:sjr:not-satisfiable}
    In the temporal election $E \in \calE_{=1}$ given below
    no outcome provides sJR.

\begin{table}[H]
    \setlength{\tabcolsep}{3pt}
    \renewcommand{\arraystretch}{1.2}
    \small
    \begin{center}
        \begin{tabular}{c | cccc cccc cccc }
        \toprule
        & \multicolumn{12}{c}{Voter} \\
        Rounds & 1 & 2 & 3 & 4 & 5 & 6 & 7 & 8 & 9 & 10 & 11 & 12 \\
        \midrule
        1 & $x$ & $x$ & $x$ & $x$ & $y$ & $y$ & $y$ & $y$ & $z$ & $z$ & $z$ & $z$ \\
        2--6 & $c_1$ & $c_2$ & $c_3$ & $c_4$ & $c_5$ & $c_6$ & $c_7$ & $c_8$ & $c_9$ & $c_{10}$ & $c_{11}$ & $c_{12}$\\
        \bottomrule
        \end{tabular}
    \end{center}
    \label{table:sJR:not-satisfiable}
\end{table}
\end{restatable}
\begin{proof}
The table specifies an election $E=(C,N,\ell,A)$ with $\ell = 6$ rounds, a set of voters $N=[12]$, and a set of candidates $C = \{x, y, z, c_1,\dots,c_{12}\}$, 
where each voter approves a single candidate in each round.

Intuitively, in the first round the voters form three cohesive blocks of size four, and in the next five rounds, every voter supports a different candidate.
Without loss of generality, assume that candidate $x$ is selected in the first round.
Let $S$ be a size-2 subset of $\{5, 6, 7,8\}$. As $\floor{\ell \cdot |S|/n} = \floor{6 \cdot 2 /12} = 1$ and voters in $S$ agree on $y$ in the first round, sJR requires that at least one voter in $S$ has positive satisfaction. Thus, at least three voters in $\{5,6, 7, 8\}$ must have positive satisfaction. By the same argument, 
at least three voters in $\{9,10,11,12\}$ must have positive satisfaction. Thus, we need to provide positive satisfaction to $6$ voters in rounds 2--6, which is clearly impossible, as in each of these rounds we can satisfy at most one voter.
Therefore, no outcome can provide sJR.
\end{proof}

The reason why sJR (and hence sEJR+) 
is hard to satisfy is that it can give strong guarantees to groups of voters
based on agreement in just a few rounds.
These rounds may be heavily contested, and selecting a candidate to satisfy one group might make it impossible to satisfy another group that agrees over the same rounds (such as groups 1--4, 5--8 and 9--12 in the proof of \Cref{prop:sjr:not-satisfiable}). In contrast, in the multiwinner setting,
all $k$ committee slots are equivalent.

Now, the difference between sEJR and EJR is that the satisfaction guarantee to a group $S$ that agrees in $t$ rounds is $\floor{\ell \cdot |S|/n}$ under to the former and $\floor{t \cdot |S|/n}$ under the latter. Can we modify the definition of sEJR+ in a similar way to obtain an axiom that can always be satisfied?
The challenge is that the multiwinner variant of this axiom, 
EJR$^\textit{\,mw}$+, provides guarantees to voters in $S$ as long as they agree on one candidate, i.e., there is no `$t$' that can be used to replace $\ell$. We would like to preserve this feature, but define the satisfaction guarantee so that it only takes into account the rounds with some level of agreement.
To this end, we replace the size of the group $|S|$ with a more `local' measure, namely, the number of voters in $S$ who agree on a candidate in a given round, 
and focus on the rounds where this quantity is high.

\begin{definition}\label{def:ejr+}
Given a temporal election $E=(C, N, \ell, A)$ and $\sigma \in [n]$, $\tau \in [\ell]$, we say that a subset of voters $S$ is {\em $(\sigma, \tau)$-cohesive} if there exists a set of $\tau$ rounds $R\subseteq [\ell]$ and a suboutcome $\outcome_R=(o_r)_{r\in R}$ such that for each $r\in R$ at least $\sigma$ voters in $S$ approve $o_r$. 
We say that an outcome $\outcome$ provides
\emph{extended justified representation\,+  (EJR+)} if for all $\sigma\in [n]$, $\tau\in [\ell]$, every $(\sigma, \tau)$-cohesive subset of voters $S$ and every round $r\in [\ell]$ with $\bigcap_{i\in S}a_{i, r}\neq\varnothing$ it holds that
\begin{equation*}
    \text{(i) } \ \sat_i(\outcome) \ge \floor{\tau \cdot \sigma/n} \text{ for some } i \in S, \quad \text{or} \quad \text{(ii) } \ o_r \in \textstyle{\bigcap_{i \in S} a_{i,r}}.
\end{equation*}
\end{definition}

In the election from \Cref{prop:sjr:not-satisfiable}, 
the group $S=\{1, 2, 3, 4\}$ is $(4, 1)$-cohesive 
(as all members of $S$ approve $x$ in round~1) and $(1, 6)$-cohesive (as at least one member of $S$ approves $(x, c_1, \dots, c_1)$), but $4\cdot 1< 12$, $1\cdot 6 < 12$, so EJR+ does not offer any representation guarantees to $S$.

We will now argue that this is the `correct' definition of EJR+ for the temporal setting. Our argument is threefold: we show that 
(a) EJR+ implies EJR and is implied by sEJR+,  (b)~EJR+ is polynomial-time verifiable, and (c) an EJR+ outcome always exists and can be found in polynomial time.

\begin{restatable}{proposition}{EJRplusImplications}\label{prop:ejr+:implications}
    It holds that: (i) sEJR+ $\implies$ EJR+ and (ii) EJR+ $\implies$ EJR. 
\end{restatable}

\begin{restatable}{proposition}{EJRplusPolyTime}\label{prop:ejr+:poly-time}
    Given a temporal election $E$ and an outcome $\outcome$,
    it can be checked in polynomial time
    whether $\outcome$ provides EJR+.
\end{restatable}
\begin{proof}
    Consider an arbitrary temporal election $E = (C, N, \ell, A)$
    and an outcome $\outcome$.
    Our goal is to check if there are integers 
    $\sigma \in [n]$, $\tau \in [\ell]$, a $(\sigma, \tau)$-cohesive subset of voters $S$, and a round $r\in [\ell]$ with $\bigcap_{i\in S}a_{i, r}\neq\varnothing$ such that $o_r\not\in \bigcap_{i\in S}a_{i, r}$ and for each $i\in S$ it holds that $\sat_i(S)<\floor{\tau \cdot \sigma/n}$. Our algorithm makes use of the observation that every superset of a $(\sigma, \tau)$-cohesive subset of voters is itself $(\sigma, \tau)$-cohesive.

    We go over all possibilities for $r\in [\ell]$ and $c\in C\setminus\{o_r\}$. For a given pair $(r, c)$, we go over all $\lambda\in [\ell]$ and identify the set of all voters $i\in N$
    who (i) approve $c$ in round $r$ 
    and (ii) have $\sat_i(\outcome)<\lambda$; denote this set by $S_{r, c, \lambda}$.
    If all voters in $S_{r, c, \lambda}$ approve $o_r$ in round $r$, we disregard this set. Otherwise,  
    for each $q\in [\ell]$ let $c_q$ be a candidate that receives the maximum number of approvals from voters in $S_{r, c, \lambda}$, and let $\alpha_q$ be the number of voters in $S_{r, c, \lambda}$ who approve $c_q$ in round $q$. We then sort $(\alpha_q)_{q\in [\ell]}$ is non-increasing order; denote the resulting list by $(\alpha'_q)_{q\in [\ell]}$.
    Finally, we check whether there exists a $q\in [\ell]$ such that 
    $\floor{q\cdot \alpha'_q/n}\ge \lambda$. If this is the case, the respective set of $q$ rounds witnesses that $S_{r, c, \lambda}$ is $(\alpha'_q, q)$-cohesive. Moreover, voters in $S_{r, c, \lambda}$ approve $c$ in round $r$ (so $\bigcap_{i\in S}a_{i, r}\neq\varnothing$, but $o_r\not\in \bigcap_{i\in S}a_{i, r}$), yet for each voter 
    $i\in S_{r, c, \lambda}$ its satisfaction is less than $\lambda\le \floor{q\cdot \alpha'_q/n}$, i.e., $S_{r, c, \lambda}$ witnesses that $\outcome$ violates EJR+. In this case, our algorithm reports that $\outcome$ does not provide EJR+.   

    Our algorithm goes over $m\ell^2$ triples $(r, c, \lambda)$, and performs
    $O(nm\ell+\ell\log \ell)$ operations for each such triple. Hence, it runs in time polynomial in the input size.

    Clearly, if our algorithm reports that EJR+ is violated, it provides an explicit witness. It remains to argue that if EJR+ is violated, our algorithm will be able to detect this. To see this, suppose that the violation of EJR+ is witnessed by a $(\sigma, \tau)$-cohesive set $S$ such that $\bigcap_{i\in S}a_{i, r}\neq\varnothing$, but $o_r\not\in \bigcap_{i\in S}a_{i, r}$, and 
     the satisfaction of each voter in $S$ is less than $\floor{\sigma \cdot \tau/n}$. Then there exists a candidate $c\in \bigcap_{i\in S}a_{i, r}$; note that $c\neq r$.
     Let $R$ be the set of rounds witnessing that $S$ is $(\sigma, \tau)$-cohesive. When our algorithm considers 
    the triple $(r, c, \lambda)$ with $\lambda=\floor{\sigma\cdot \tau/n}$, it constructs the set $S_{r, c, \lambda}$ and by definition it holds that $S\subseteq S_{r, c, \lambda}$. Then, for each $q\in R$ we have $\alpha_q\ge \sigma$, since at least $\sigma$ voters from $S$ approve the same candidate in round $q$. Hence, the sequence $(\alpha_q)_{q\in [\ell]}$ contains at least $|R|=\tau$ values that are 
    at least $\sigma$, and consequently our algorithm will be able to identify that EJR+ is violated.
\end{proof}
To prove that every temporal election has an EJR+ outcome, we show that the temporal variant of the 
$\varepsilon$-lsPAV rule (introduced by \citet{aziz2018complexityepjr} in the multiwinner setting and adapted to the temporal setting by \citet{chandak2024proportional}) with $\varepsilon=1/(2\ell^2)$
always outputs EJR+ outcomes and runs in polynomial time. 

\begin{definition}\label{def:ls-pav}
For a temporal election $E = (C,N,\ell,A)$,
the \emph{harmonic score} of an outcome $\outcome$ is defined as 
\(
    s_{H}(\outcome,E) =
    \sum_{i \in N}\sum_{j = 1}^{\sat_i(\outcome)}\frac{1}{j}.
\)
The $\varepsilon$-lsPAV rule outputs all outcomes
$\outcome$ in $C^\ell$ such that for every outcome $\outcome' \in C^\ell$
that agrees with $\outcome$ in $\ell-1$ rounds it holds that
$s_{H}(\outcome,E) + \varepsilon \ge s_{H}(\outcome',E)$.
\end{definition}

Setting $\varepsilon \ge 1/(2\ell^2)$ ensures that $\varepsilon$-lsPAV outputs an outcome in polynomial time (see~\citet{aziz2018complexityepjr} and \citet{chandak2024proportional}). Moreover, we show that $\varepsilon$-lsPAV provides EJR+ if we set $\varepsilon < 1/\ell^2$.

\begin{restatable}{theorem}{LSPAVsatisfiesEJRplus}\label{thm:lspav:ejr+}
    For every temporal election $E = (C, N, \ell, A)$, each output of $\varepsilon$-lsPAV with $\varepsilon < 1/\ell^2$ provides EJR+.
\end{restatable}
\begin{proof}
    Fix an arbitrary temporal election
    $E = (C, N, \ell, A)$,
    and let $\mathbf{o}$ be an output of $\varepsilon$-lsPAV on $E$.
    
    For a contradiction, assume that $\outcome$ does not provide EJR+.
    Then there exist $\sigma\in [n]$, $\tau\in [\ell]$, 
    a $(\sigma, \tau)$-cohesive subset of voters $S \subseteq N$, a round $r \in [\ell]$, and 
    a candidate $c \in C$ that is approved by all voters in $S$ in round $r$
    such that $c \neq o_r$,
    and for every $i \in S$
    it holds that $\sat_i(\mathbf{o}) < \lfloor \tau \cdot \sigma/n\rfloor$.
    If $\sigma = |S|$, this immediately implies that EJR is violated, which is a contradiction, 
    as \citet[Theorem 4.5]{chandak2024proportional} show that $\varepsilon$-lsPAV provides EJR for $\varepsilon>1/\ell^2$.
    Thus, let us assume that $\sigma < |S|$.

    Let $R$ be a set of $\tau$ rounds witnessing that $S$ is $(\sigma, \tau)$-cohesive, i.e., in each round $q\in R$ there are $\sigma$ voters in $S$ that approve a common candidate. 
    We can assume that $r\in R$, with $c_r=c$.
    For each $q\in R$, let $c_q$ be 
    the respective candidate and let $S_q$ be the subset of voters in $S$ who approve $c_q$ in round $q$; note that $|S_q|\ge \sigma$. 
    For each $q\in R$, let $\Delta(o_q, c_q)$ be the change in the harmonic score obtained by replacing $o_q$ with $c_q$ in round $q$.
    We have
    \begin{align*}
        \Delta(o_q,c_q) &=
            \sum_{\substack{i \in N \\ 
                o_q \not\in a_{i,q}, c_q \in a_{i,q}}}
            \frac{1}{\sat_i(\mathbf{o})+1} -
            \sum_{\substack{i \in N \\
                o_q \in a_{i,q}, c_q \not\in a_{i,q}}}
            \frac{1}{\sat_i(\mathbf{o})}\\
        &\ge
            \sum_{\substack{i \in S_q \\
                o_q \not\in a_{i,q}}}
            \frac{1}{\sat_i(\mathbf{o})+1} -
            \sum_{\substack{i \in N \setminus S_q \\
                o_q \in a_{i,q}, c_q \not\in a_{i,q}}}
            \frac{1}{\sat_i(\mathbf{o})}
        \tag{for every $i \in S_q$ we have $c_q \in a_{i,q}$}\\
        &\ge
            \sum_{\substack{i \in S_q \\
                o_q \not\in a_{i,q}}}
            \frac{1}{\sat_i(\mathbf{o})+1} -
            \sum_{\substack{i \in N \setminus S_q \\
                o_q \in a_{i,q}}}
            \frac{1}{\sat_i(\mathbf{o})}
        \tag{remove a condition in the second sum}\\
        &=
            \sum_{i \in S_q}
                \frac{1}{\sat_i(\mathbf{o})+1} -
            \sum_{\substack{i \in N \setminus S_q\\ o_q \in a_{i,q}}}
                \frac{1}{\sat_i(\mathbf{o})} -
            \sum_{\substack{i \in S_q\\ o_q \in a_{i,q}}}
                \frac{1}{\sat_i(\mathbf{o}) + 1}
        \tag{add and subtract the same}\\
        &=
            \sum_{i \in S_q}
                \frac{1}{\sat_i(\mathbf{o})+1} -
            \sum_{\substack{i \in N\\ o_q \in a_{i,q}}}
                \frac{1}{\sat_i(\mathbf{o})} +
            \sum_{\substack{i \in S_q\\ o_q \in a_{i,q}}}
                \left(\frac{1}{\sat_i(\mathbf{o})} -
                    \frac{1}{\sat_i(\mathbf{o}) + 1}\right)
        \tag{add and subtract the same}\\
        &\ge
            \sum_{i \in S_q}
                \frac{1}{\sat_i(\mathbf{o})+1} -
            \sum_{\substack{i \in N\\ o_q \in a_{i,q}}}
                \frac{1}{\sat_i(\mathbf{o})}
        \tag{remove a non-negative term}\\
        &\ge
            \frac{\sigma}{\lfloor \tau \cdot \sigma / n \rfloor} -
        \sum_{\substack{i \in N \\ o_q \in a_{i,q}}}
            \frac{1}{\sat_i(\mathbf{o})}.
        \tag{$|S_q| \ge \sigma$ and $\sat_i(\outcome) < \floor{\tau \cdot \sigma/n}$ for $i \in S_q \subseteq S$}
    \end{align*}
    For $q=r$ and $c_r=c$, we have $|S_r| =|S|\ge \sigma + 1$, so we can strengthen the last transition by replacing $\sigma$ with $\sigma+1$. We obtain
    \[
        \Delta(o_r,c) \ge \frac{\sigma + 1}{\lfloor \tau \cdot \sigma / n \rfloor} -
        \sum_{\substack{i \in N \\ o_r \in a_{i,r}}}
            \frac{1}{\sat_i(\mathbf{o})}.
    \]
    Next, we sum $\Delta(o_q,c_q)$ over all $q\in R$,
    including round $r$:
    \begin{align*}
        \sum_{q \in R} \Delta(o_q,c_q) &\ge
            \frac{\tau \cdot \sigma + 1}{\lfloor \tau \cdot \sigma / n \rfloor} -
            \sum_{q \in R} \sum_{\substack{i \in N \\ o_q \in a_{i,q}}}
                \frac{1}{\sat_i(\mathbf{o})}\\
        &\ge
            \frac{\tau \cdot \sigma + 1}{\lfloor \tau \cdot \sigma / n \rfloor} -
            \sum_{q \in R} \sum_{\substack{i \in N \\ o_q \in a_{i,q}}}
                \frac{1}{\sat_i(\mathbf{o}_R)}
        \tag{$\sat_i(\mathbf{o}) \ge \sat_i(\mathbf{o}_R)$}\\
        &=
            \frac{\tau \cdot \sigma + 1}{\lfloor \tau \cdot \sigma / n \rfloor} -
            \sum_{i \in N} \sum_{\substack{q \in R \\ o_q \in a_{i,q}}}
            \frac{1}{\sat_i(\mathbf{o}_R)}
        \tag{change the order of the summation}\\
        &=
            \frac{\tau \cdot \sigma + 1}{\lfloor \tau \cdot \sigma / n \rfloor} -
            \sum_{\substack{i \in N\\o_q\in a_{i, q}\text{ for some }q\in R}} 1
        \tag{by the definition of $\sat_i(\outcome_R)$}\\
        &\ge  \frac{\tau \cdot \sigma + 1}{\lfloor \tau \cdot \sigma / n \rfloor} - n\\
        &\ge  n + \frac{n}{\tau \cdot \sigma} - n\\
        &=  \frac{n}{\tau \cdot \sigma}.
    \end{align*}
        Thus, by the pigeonhole principle, there exists a round $q \in R$ in which switching from $o_q$ to $c_q$ would increase the PAV score by at least $n/(\tau^2\sigma)$.
        Since $n\ge \sigma$ and $\tau\le \ell$, we have
        $n/(\tau^2\sigma) \ge 1/\tau^2 \ge 1/\ell^2$, which contradicts the assumption that $\outcome$ was an outcome of $\varepsilon$-lsPAV with $\varepsilon < 1/\ell^2$.
\end{proof}

We note that $\varepsilon$-lsPAV always produces 
an outcome that satisfies condition~(i) in \Cref{def:ejr+}. Thus, we can obtain a stronger, but still satisfiable version of EJR+ by removing condition~(ii).
However, we decided to keep it for consistency with the multiwinner version of this axiom, i.e.,  EJR$^\textit{\,mw}$+.\footnote{This strengthening of EJR+  appears to be conceptually different from EJR$^\textit{\,mw}$+.
Indeed, one might wonder if this idea can be ported back to the multiwinner setting, 
resulting in an axiom that is distinct from EJR$^\textit{\,mw}$+.
This said, we feel that
\Cref{def:ejr+} is the `correct' adaptation of EJR$^\textit{\,mw}$+
to the temporal setting, as it has the same desirable features---satisfiability, efficient computability, and efficient verifiability---as EJR$^\textit{\,mw}$+.}

Finally, let us define a weaker version of the EJR+ axiom in the spirit of wEJR/wPJR/wJR formulations from \Cref{def:wJR:wPJR:wEJR:combined}.

\begin{definition}\label{def:wejr+}
Given a temporal election $E=(C, N, \ell, A)$, 
we say that an outcome $\outcome$ provides
\emph{weak extended justified representation\,+  (wEJR+)} if for all $\sigma\in [n]$, every $(\sigma, \ell)$-cohesive subset of voters $S$ and every round $r\in [\ell]$ with $\bigcap_{i\in S}a_{i, r}\neq\varnothing$ it holds that
\begin{itemize}
        \item[(i)] $\sat_i(\outcome) \ge \floor{\ell \cdot \sigma/n}$ for some $i\in S$, or
        \item[(ii)] $o_r \in \bigcap_{i \in S} a_{i,r}$.
    \end{itemize}
\end{definition}

We then show that wEJR+ as defined above fits well into our axiomatic hierarchy.

\begin{restatable}{proposition}{WEJRplusrelationship} \label{prop:WEJRplusrelationship}
   It holds that: (i) EJR+ $\implies$ wEJR+ and (ii) wEJR+ $\implies$ wEJR.
\end{restatable}

\section{Full Justified Representation} \label{sec:fjr}

The next axiom that we would like to adapt to the temporal setting is
\emph{full justified representation} (FJR$^\textit{\,mw}$) \citep{peters2021fjr}.
Similarly to EJR$^\textit{\,mw}$, it aims to select a size-$k$ committee $W$ that provides guarantees to every subset of voters $S$ in proportion to its size $|S|$ and the level of agreement. However, when measuring agreement, 
it does not require all voters in $S$ to approve
the same set of candidates.
Instead, voters in $S$ go over all subsets of candidates $T$ with $|T|\le k\cdot |S|/n$ and compute their satisfaction from $T$ as $\kappa^\textit{\,mw}_S(T)=\min_{i\in S}|T\cap a_i|$, where $a_i$ is the approval ballot of voter $i$.
FJR$^\textit{\,mw}$ then requires that 
for each choice of $T$,
some voter in $S$ approves at least $\kappa^\textit{\,mw}_S(T)$ members of~$W$.

FJR$^\textit{\,mw}$ is (strictly) stronger than EJR$^\textit{\,mw}$: intuitively, it provides EJR-like guarantees to a larger collection of voter subsets, so every outcome that 
provides FJR$^\textit{\,mw}$ also provides 
EJR$^\textit{\,mw}$, but the converse is not true.
Moreover, every election
has an FJR$^\textit{\,mw}$ outcome \citep{peters2021fjr}.

To translate FJR$^\textit{\,mw}$ to the temporal setting, a natural approach is to iterate over subsets of rounds $R$ (and the best possible selection of outcomes for these rounds) instead of subsets of candidates $T$.
This yields the following definition.

\begin{definition}\label{def:sFJR}
    Given a temporal election
    $E =(C, N, \ell, A)$, we say that
    an outcome $\outcome$ provides \emph{strong full justified representation} (sFJR)
    if for every subset of voters $S \subseteq N$,
    every subset of rounds $R \subseteq [\ell]$
    such that $|R| = \floor{\ell \cdot |S|/n}$, 
    and every outcome $\outcome'$,
    there is a voter $i \in S$ such that
    \(
        \sat_i(\outcome) \ge \min_{j \in S} \sat_j(\outcome'_R).
    \)
    Equivalently, for every subset of voters $S$
    there is an $i \in S$ such that
    \[
        \sat_i(\outcome) \ge 
        \max_{R \subseteq [\ell]:|R|= \floor{\ell \cdot |S|/n}} \
        \max_{\outcome'_R \in C^R} \
        \min_{j\in S} \
        \sat_j(\outcome'_R).
    \]
\end{definition}

However, just as for sEJR+, the direct approach results in a definition that is too demanding:
later in this section, we will show that
sFJR implies sEJR, 
which means that some elections have no sFJR outcomes. A simple way to modify this definition
is to consider the worst (rather than the best) subset of rounds $R$ for the set $S$ (while still allowing the voters in $S$ to choose outcomes for rounds in $R$). 

\begin{definition}\label{def:wFJR}
    Given a temporal election
    $E =(C, N, \ell, A)$, 
    an outcome $\outcome$ provides \emph{weak full justified representation (wFJR)}
    if for every subset of voters $S \subseteq N$
    there is a subset of rounds $R \subseteq [\ell]$ of size $|R| = \floor{\ell \cdot |S|/n}$
    such that for every other outcome $\outcome'$
    there is a voter $i \in S$ with
    \(
        \sat_i(\outcome) \ge \min_{j \in S} \sat_j(\outcome'_R).
    \)
    Equivalently, for every subset of voters $S$,
    there is an $i \in S$ such that
    \[
        \sat_i(\outcome) \ge 
        \min_{R \subseteq [\ell]:|R|= \floor{\ell \cdot |S|/n}} \
        \max_{\outcome'_R \in C^R} \
        \min_{j\in S} \
        \sat_j(\outcome'_R).
    \]
\end{definition}

The reason we placed the concept introduced in \Cref{def:wFJR}
at the `weak' level of our hierarchy is that requiring 
a way of jointly reaching a satisfaction threshold (via $\outcome'_R$) in {\em all} subsets of rounds of size $\floor{\ell \cdot |S|/n}$ is similar in spirit, but weaker than requiring total agreement in all rounds.
Indeed, in a moment we will show that wFJR implies wEJR.

Finally, for an FJR axiom in the spirit of EJR/PJR/JR from \Cref{def:JR:PJR:EJR:combined},
we would like to move away from considering all (subsets of) rounds. To this end, when determining what $S$ `deserves', we allow $S$ to designate a subset of rounds $T$, and then we pick the worst subset of rounds $R$, 
with size of $R$ dependent on $|T|$ (rather than $\ell$).
This way,
if the voters in $S$ can reach an acceptable compromise
over a large set of rounds $T$, they obtain 
strong satisfaction guarantees; however, if $T$
is small, they still get a (weaker) guarantee.

\begin{definition}
\label{def:FJR}
    Given a temporal election
    $E =(C, N, \ell, A)$, we say that
    an outcome $\outcome$ provides \emph{full justified representation (FJR)}
    if for every subset of voters $S \subseteq N$
    and every subset of rounds $T \subseteq [\ell]$
    there exists a subset $R \subseteq T$
    with $|R| = \floor{|T| \cdot |S|/n}$
    such that for every other outcome $\outcome'$
    there is a voter $i \in S$ with
    \(
        \sat_i(\outcome) \ge \min_{j \in S} \sat_j(\outcome'_R).
    \)
    Equivalently, for every subset $S\subseteq N$
    there is a voter $i \in S$ with $\sat_i(\outcome) \ge \max_{T \subseteq [\ell]}\mu_S(T)$, where
    \[
        \mu_S(T) =  
        \
        \min_{R \subseteq T:|R|= \floor{|T| \cdot |S|/n}} \
        \max_{\outcome'_R \in C^R} \
        \min_{j\in S} \
        \sat_j(\outcome'_R).
    \]
\end{definition}

The first argument that our definitions are `correct'
is that all of the expected implications hold.

\begin{samepage}
\begin{restatable}{proposition}{PropFJRimplications}\label{prop:fjr_implications}
    It holds that: (i) sFJR $\implies$ sEJR and FJR, (ii) FJR $\implies$ EJR and wFJR, and (iii) wFJR $\implies$ wEJR.
\end{restatable}
\end{samepage}

\begin{algorithm}[!t] 
    \caption{Greedy Cohesive Rule}
    \label{alg:gcr}
    \begin{algorithmic}[1]{
        \REQUIRE A temporal election $E = (C, N, \ell, A)$
        \STATE $\outcome \leftarrow (c,\dots,c)$ for arbitrary $c \in C$
        \STATE $V \leftarrow N$, \quad 
            $p \leftarrow 1$
        \WHILE{$V \neq \varnothing$}
            \STATE pick $S_p$ from $\arg\max_{S\subseteq V}\max_{T\subseteq [\ell]}\mu_S(T)$ 
            \label{alg:line_fjr_req}
            \STATE pick $T_p$ from $\arg\max_{T\subseteq[\ell]}\mu_{S_p}(T)$
            \STATE $V \leftarrow V \setminus S_p$, \quad
                $p \leftarrow p+1$
        \ENDWHILE
        \STATE $\pi \leftarrow$ a permutation of $[p]$ such that
        $|T_{\pi(1)}| \le \dots \le |T_{\pi(p)}|$ \label{alg:line_perm}
        \STATE $\mathcal{T} \leftarrow [\ell]$
        \FOR{$q \in [p]$}
            \STATE $i \leftarrow \pi(q)$
            \STATE $R \leftarrow$ arbitrary subset of $\floor{|T_i|\cdot|S_i|/n}$ rounds from $T_i \cap \mathcal{T}$
            \STATE pick $\outcome'_R$ from $\arg\max_{\outcome''_R\in C^{|R|}}\min_{j \in S_i} \sat_j(\outcome''_R)$
            \STATE $\outcome_R \leftarrow \outcome'_R$, \quad $\mathcal{T} \leftarrow \mathcal{T}\setminus R$
        \ENDFOR
        \RETURN $\outcome$ }
        \end{algorithmic}
\end{algorithm}

In the multiwinner setting, \citet{peters2021fjr}
show that the Greedy Cohesive Rule (GCR) \citep{bredereckF0N19} always outputs
FJR$^\textit{\,mw}$ outcomes.
We will now show that a variant of GCR satisfies FJR in the temporal setting, providing further evidence that our definition of FJR is `correct'.
Our starting point is an adaptation of GCR to the temporal setting by \citet{elkind2025verifying}, who show that their variant of the rule satisfies EJR.
We modify the algorithm of \citet{elkind2025verifying}, and show that our version satisfies FJR.

A formal description of GCR
is presented in \Cref{alg:gcr}.
Intuitively, during the first stage (lines 1--7) GCR partitions the voters into sets $S_1,S_2,\dots,S_p$: at each step $q \in [p]$ it picks a subset of voters $S_q$ with the  highest `demand' $\max_T\mu_S(T)$, adds it to the partition, and discards all voters in $S_q$.
During the second stage (lines 8--16), it iterates through the sets $S_1,S_2,\dots,S_p$
in order of the size of the associated
subset of rounds $T$, from smallest to largest; when processing $S_q,$ it selects outcomes for some of the rounds in $T_q$ so as to satisfy the `demand' of voters in $S_q$.

\begin{restatable}{theorem}{thmfjrgcr} \label{thm:fjr_gcr}
    For every temporal election $E = (C, N, \ell, A)$,
    every output of GCR provides FJR.
\end{restatable}
\begin{proof}
    By construction, the sets $S_1,S_2,\dots,S_p$ are pairwise disjoint: for each $q \in [p]$, when we select $S_q$, we remove voters in $S_q$ from $V$, and therefore no set $S_i$ with $i > q$ can contain a voter from $S_q$.
    We will now argue that during the second stage of the algorithm (lines 8--16), when we consider the set $S_i$, the set of rounds $T_i$ contains $\floor{|T_i| \cdot |S_i|/n}$ rounds whose outcomes have not been set in previous iterations. For readability, we will assume that $\pi(q) = q$ for all $q \in [p]$.

    Suppose for a contradiction that this is not the case, and let $q \in [p]$ be the first index such that there are fewer than $\floor{|T_q| \cdot |S_q|/n}$ rounds from $T_q$ available.
    This means that strictly more than $|T_q| - \floor{|T_q| \cdot |S_q|/n}$ of the rounds in $T_q$ have been taken up in previous iterations, and hence
    \begin{equation*}
        \sum_{i=1}^{q} \left\lfloor\frac{|T_i| \cdot |S_i|}{n}\right\rfloor > |T_q|.
    \end{equation*}
    Since $|T_q| \geq |T_i|$ for all $i \leq q$, this inequality implies
    \begin{align*}
        \frac{|T_q|}{n} \cdot \sum_{i=1}^q |S_i|
        & \geq \sum_{i=1}^q \frac{|T_i| \cdot |S_i|}{n} 
         \ge \sum_{i=1}^q \left\lfloor\frac{|T_i| \cdot |S_i|}{n}\right\rfloor
         > |T_q|.
    \end{align*}
    Hence, $\sum_{i=1}^q |S_i| > n$, a contradiction with the fact that the sets $S_1,\dots,S_q$ are pairwise disjoint.

     Thus, when the algorithm processes $S_i$, it is presented with $\floor{|T_i|\cdot |S_i|/n}$ rounds from $T_i$ and selects outcomes for these rounds so as to maximize the minimum satisfaction of voters in $S_i$ (line~13). Hence, the satisfaction of every $j\in S_i$ from $\outcome$ is at least $\mu_{S_i}(T_i)$.
     By the choice of $T_i$ it means that $j$'s satisfaction is at least $\max_{T\subseteq [\ell]} \mu_{S_i}(T)$, 
     i.e., the FJR condition is satisfied for all voters in $S_i$. In particular, none of the sets $S_1,\dots,S_p$ can be a witness that FJR is violated.
    
    It remains to show that no subset of voters $S\neq S_1, \dots,S_p$ can witness that FJR is violated. Fix a subset $S \subseteq V$.
    Since $S_1,\dots,S_p$ form a partition of $N$, the set $S$ must intersect with at least one of them. 
    Let $i=\min\{q: S_q\cap S\neq\varnothing\}$, and let $j$ be some voter in $S\cap S_i$. As $j\in S_i$, 
    her satisfaction from $\outcome$ is at least $\max_{T\subseteq [\ell]}\mu_{S_i}(T)$. On the other hand,
    since \Cref{alg:gcr} picked $S_i$ over $S$ in line~4, we have $\max_{T'\subseteq [\ell]}\mu_{S}(T')\le \max_{T'\subseteq [\ell]}\mu_{S_i}(T')$. 
    Hence, $j$'s satisfaction from $\outcome$ is at least $\max_{T'\subseteq [\ell]}\mu_{S}(T')$. As $j\in S$,  
    the set $S$ cannot be a witness that FJR is violated.
\end{proof}

The running time of GCR is not polynomial in the input size, so it remains open whether an FJR outcome can be computed in polynomial time; this problem is also open in the multiwinner case.

\begin{remark}
The fact that there always exists an outcome that provides FJR is also implied by Theorem 11 of \citet{MasPieSko-2023-GeneralProportionality}, because our notion of FJR is implied by their analogous notion in a more general setting. We discuss this connection in detail in \Cref{app:general-prop}.
\end{remark}

\section{Full Proportional Justified Representation}\label{sec:fpjr}

In the definition of FJR, the guarantee provided to each group $S$ is of the form `some voter in $S$ has high satisfaction', i.e., it is a lower bound on $\max_{i\in S}\sat_i(\outcome)$; this approach is also used in the definition of EJR. In contrast, 
the guarantee provided by PJR is of the form `collectively, voters in $S$ have high satisfaction', i.e., it is a lower bound on $\sat_S(\outcome)$. One can combine the FJR approach to deciding what each group deserves with a PJR-style 
collective guarantee; we refer to the resulting notion as FPJR (the multiwinner version of FPJR was very recently proposed by \citet{kalayci2025fpjr}).
Based on the discussion in Sections~\ref{sec:ejr-plus} and~\ref{sec:fjr}, we define strong FPJR, FPJR, and weak FPJR.
\begin{definition}\label{def:fpjr}
    Given a temporal election $E =(C,N , \ell, A )$,
    an outcome $\outcome$
    provides
    \emph{strong full proportional justified representation (sFPJR)}
    (resp., \emph{full proportional justified representation (FPJR)}, or
    \emph{weak full proportional justified representation (wFPJR)})
    if for every $S \subseteq N$
    we have $\sat_S(\outcome) \ge \rho^s(S)$
    (resp., $\sat_S(\outcome) \ge \rho(S)$, or
    $\sat_S(\outcome) \ge \rho^w(S)$),
    where
    \begin{align*}
    \rho^s(S) &= \displaystyle
        \max_{R \subseteq [\ell]:|R|= \floor{\ell \cdot |S|/n}} \
        \max_{\outcome'_R \in C^R} \
        \min_{i\in S} \
        \sat_i(\outcome'_R), \\
    \rho(S) &= \displaystyle 
            \max_{T \subseteq [\ell]} \
            \min_{R \subseteq T:|R|=\floor{|T| \cdot |S|/n}} \
            \max_{\outcome'_R \in C^R} \
            \min_{i \in S} \
            \sat_i(\outcome'_R), \\
    \rho^w(S) &= \displaystyle
            \min_{R \subseteq [\ell]:|R|=\floor{\ell \cdot |S|/n}} \
            \max_{\outcome'_R \in C^R} \
            \min_{i\in S} \
            \sat_i(\outcome'_R).
    \end{align*}
\end{definition}

Similarly to EJR,
FPJR is implied by FJR and implies PJR
on each level of our axiomatic hierarchy.

\begin{restatable}{proposition}{PropFPJRimplications}\label{prop:fpjr:implications}
    It holds that: (i) sFJR $\implies$ sFPJR $\implies$ sPJR and FPJR, (ii) FJR $\implies$ FPJR $\implies$ PJR and wFPJR, and (iii) wFJR $\implies$ wFPJR $\implies$ wPJR.
\end{restatable}
Propositions~\ref{prop:fpjr:implications} and~\ref{prop:sjr:not-satisfiable} imply that some elections have no sFPJR outcomes.
However, sFPJR turns out to be satisfiable 
for elections in $\calE_{\ge 1}$ as long as the number of rounds $\ell$ 
is divisible by the number of voters $n$.%
\footnote{This turns out to be the best we can hope for in terms of satisfiability of the strong versions of our axioms since for sEJR, and hence every axiom that implies it, even in this restricted case there are elections where it is unsatisfiable. We prove this in \Cref{app:fpjr}. \label{footnote_fpjr}}
Specifically, we show that 
in this case the simple \emph{Serial Dictatorship Rule} (SDR)
\citep{satterthwaite1981strategy}
outputs sFPJR outcomes.

SDR operates by iterating through the rounds
in $[\ell]$ while cycling through voters in $N$;
in each round $r \in [\ell]$
it picks an outcome approved by the current voter, i.e., it selects an arbitrary
$o_r \in a_{r\!\!\!\mod\! n,r}$
(for elections in $\calE_{\ge 1}$ the set $a_{r\!\!\!\mod\! n,r}$ cannot be empty, so such $o_r$ has to exist).

\begin{theorem} \label{thm:sfpjr-sdr}
    For every $n$-voter $\ell$-round temporal election $E \in \calE_{\ge 1}$
    such that $n|\ell$,
    every output of SDR provides sFPJR.
\end{theorem}
\begin{proof}
    Consider an election $E = (C, N, \ell, A)\in \calE_{\ge 1}$ with $n | \ell$.
    Let SDR return the outcome $\outcome$, and fix a subset of voters $S$.
    Since each voter selects the outcome of exactly $\ell/n$ rounds, we have
    $\sat_S(\outcome) \geq \ell \cdot |S|/{n}$. Also, for every $R\subseteq [\ell]$, every $\outcome'_R$ and every $i\in N$ we have $\sat_i(\outcome'_R)\le |R|$ and hence $\rho^s(S)\le \floor{\ell\cdot |S|/n}$. Thus, $\sat_S(\outcome)\ge \rho^s(S)$, which is what we wanted to prove.
\end{proof}

Interestingly, if an election with $n|\ell$ is additionally in $\calE_{=1}$, then the outputs of SDR also provide wEJR.
This is because for 
elections from $\calE_{=1}$
we can show that wPJR
(which by \Cref{prop:fpjr:implications} is implied by sFPJR)
is equivalent to wEJR.

\begin{restatable}{proposition}{PropwPJRiffwEJR}\label{prop:wpjr_iff_wEJR}
    For every temporal election $E \in \calE_{=1}$,
    it holds that wPJR $\iff$ wEJR.
\end{restatable}

\section{Core Stability} \label{sec:core}
To complete the analysis of proportionality notions in the temporal setting, we consider the classic concept of \emph{core stability}, which originates in cooperative game theory. Core stability for multiwinner voting was defined by \citet{aziz2017jr}, and it remains open if all multiwinner elections admit outcomes in the core.
This concept captures resistance to deviations: a subset of voters can deviate by selecting outcomes in a subset of rounds whose size is proportional to the size of the group, and an outcome is stable if there is no deviation that benefits all members of the group. Strong core stability, core stability and weak core stability differ in how the voters are allowed to choose this subset of rounds.

\begin{definition}
    Given a temporal election $E =(C, N, \ell, A )$,
    an outcome $\outcome$
    provides
    \begin{itemize}[leftmargin=0.3in]
    \item 
    \emph{strong core stability (sCore)}
    if for each $S\subseteq N$, each $R\subseteq [\ell]$ with $|R|=\floor{\ell\cdot |S|/n}$ 
    and each $\outcome'_R\in C^R$ there is an $i\in S$ with $\sat_i(\outcome)\ge \sat_i(\outcome'_R)$;
    \item
    \emph{core stability (Core)}
    if for each $S\subseteq N$ and each $T\subseteq [\ell]$ 
    there is $R\subseteq T$ with 
    $|R|\!=\!\floor{|T| \cdot |S|/n}$ such that
    for each $\outcome'_R\in C^R$ there is an $i\in S$ with $\sat_i(\outcome)\ge \sat_i(\outcome'_R)$;
    \item
    \emph{weak core stability (wCore)}
    if for each $S\subseteq N$ there is an $R\subseteq [\ell]$ with $|R|=\floor{\ell\cdot |S|/n}$ such that
    for each $\outcome'_R\in C^R$ there is an $i\in S$ with $\sat_i(\outcome)\ge \sat_i(\outcome'_R)$.
    \end{itemize}
\end{definition}

In the multiwinner voting setting,
every outcome in the core provides FJR. 
This is also true in the temporal setting, 
for each level of our axiomatic hierarchy.

\begin{restatable}{proposition}{PropCoreImplications}\label{prop:core_implications}
    It holds that: (i) sCore $\implies$ sFJR and Core, (ii) Core $\implies$ FJR and wCore, and (iii) wCore $\implies$ wFJR.
\end{restatable}
The satisfiability of Core and wCore (but not sCore) remains open in the temporal setting.

\section{Separation Results}
\label{sec:independence}

Finally, we show that, apart from the implications already established, no further implication holds among the axioms in Figure~\ref{fig:axiom_map}. To rule out all remaining potential implications, we identify twelve specific pairs of axioms for which we construct outcomes that satisfy one notion but not the other.
These counterexamples exhaust all possible implication relationships among the axioms we study.

\begin{proposition} \label{prop:separation_prop}
For any two axioms $A, B$ in Figure \ref{fig:axiom_map}, $A$ implies $B$ if and only if there exists a path (a sequence of arrows) from $A$ to $B$.
\end{proposition}

\section{Conclusion}

We have explored several candidates for further strengthening all existing concepts studied in the temporal voting literature so far.
We adapted EJR+, FJR, FPJR, and the Core from multiwinner elections to the temporal setting.
We showed the polynomial-time computability of EJR+ outcomes, the existence of FJR outcomes, and the polynomial-time computability of sFPJR outcomes in a specific instance.
We also establish non-existence results for stronger variants of these concepts, including the Core, and analyze the implications (or lack thereof) between the considered axioms.
Our work provides a comprehensive analysis of proportionality concepts in temporal voting.

Possible directions for future work include exploring domain restrictions for notions where proportionality is not guaranteed, or studying extensions of these concepts in the broader setting of participatory budgeting.

\bibliographystyle{plainnat}
\bibliography{abb,bib}

\newpage

\appendix

\section{Relation to General Approval Voting}
\label{app:general-prop}

In this appendix,
we discuss the relation between our EJR and FJR axioms and
analogous notions defined by \citet{MasPieSko-2023-GeneralProportionality}
for a more general setting.

\subsection{Notation for General Approval Voting}

Let us first introduce a notation for the \emph{general approval voting} setting.
A general approval election instance is a tuple $\calI = (\calC,\calN,\calF,\calA)$,
where $\calC$ is a set of candidates,
$\calN$ is a set of $n$ voters,
$\calF \subseteq 2^\calC$ is a family of \emph{feasible} subsets of the set of candidates,
and $\calA = (A_i)_{i \in \calN}$ is the approval profile,
where $A_i \subseteq \calC$ is a ballot of each voter $i \in \calN$, i.e., a subset of candidates that he or she approves.
We assume that the family $\calF$
is closed under inclusion, i.e.,
if $Y \subseteq X \subseteq C$
and we know that $X \in \calF$,
then also $Y \in \calF$.
In particular, this means that necessarily $\varnothing \in \calF$.
Let $\bar{\calF}$ be a subset of \emph{maximal} feasible subsets, i.e.,
subsets $X \in \calF$,
for which there is no $Y \in \calF$
such that $X \subsetneq Y$.
Then, a \emph{voting rule}
is a function $f$ that for every instance
$\calI = (\calC,\calN,\calF,\calA)$
returns a non-empty subset of maximal feasible subsets of candidates
$f(I) \subseteq \bar{\calF}$.

For a subset of voters $S \subseteq \calN$
and a feasible subset of candidates $X \in \calF$,
the satisfaction of $S$ for $X$
is a number of candidates in $X$
that are approved by at least one voter in $S$, i.e.,
$\sat_S(X)=|\{c \in X : c \in \bigcup_{i \in N}A_i\}|$.
If $S = \{i\}$ for some $i \in \calN$,
we drop the brackets and write simply $\sat_i(X)$.

To show that temporal voting is a special case of general approval voting,
for every temporal election $E=(C,N,\ell,A)$
let us define a corresponding general approval voting instance
$\calI_E = (\calC_E, \calN_E, \calF_E, \calA_E)$.
The set of candidates in $\calI_E$ will be
the set of all pairs of candidates and rounds, i.e.,
$\calC_E = \{(c,r) : c \in C, r \in [\ell]\}$.
The set of voters will be identical, i.e., $\calN_E = N$.
Next, for every voter $i \in N$,
candidate $c \in C$, and round $r \in [\ell]$,
we will say that $i$ approves $(c,r)$,
if and only if, it approves candidate $c$ in round $r$, i.e.,
$\calA_E = (A_i)_{i \in N}$ where
$A_i = \{(c,r) \in \calC_E : c \in a_{i,r}\}$.
Finally, the feasible subsets of $\calC_E$ are exactly those
in which every round appears at most once, i.e.,
$\calF_E = \{X \subseteq \calC_E : ((c,r),(c',r) \in X) \Rightarrow (c = c')\}$.
Then, there is a one-to-one correspondence
between the outcomes $\outcome \in C^\ell$ in temporal election $E$
and maximal feasible subsets of $\calC_E$, i.e.,
subsets $X \in \calF_E$ in which every round appears exactly once.
In turn, all feasible subsets, not necessarily maximal,
correspond to suboutcomes $\outcome_R \in C^R$ for some subsets of rounds $R \subseteq [\ell]$, where the selected candidates in the remaining rounds are not specified
(including suboutcome $\outcome_\varnothing$,
in which we do not specify a selected candidate
in any round).
Note that for every $R \subseteq [\ell]$ suboutcome $\outcome_R$
and a corresponding feasible set $X$
we have $\sat_S(\outcome_R) = \sat_S(X)$,
for every subset of voters $S \subseteq N$.

\subsection{EJR}

Let us now discuss the relation between EJR
as defined in \Cref{def:JR:PJR:EJR:combined}
and \emph{base extended justified representation (BEJR)}
introduced by
\citet{MasPieSko-2023-GeneralProportionality}
for the general approval voting setting.
The definition of BEJR is as follows.

\begin{definition}
\label{def:bejr}
    Given a general approval election instance 
    $I = (\calC,\calN,\calF,\calA)$
    a group of voters $S \subseteq \calN$
    \emph{deserves} a satisfaction level of $\gamma$,
    if for each feasible set $X \in \calF$, it holds that either
    \begin{itemize}
        \item[i)] there exists $Y \subseteq \bigcap_{i \in S} A_i$
            with $|Y| \ge \gamma$ such that $X \cup Y \in \calF$, or
        \item[ii)] $|S|/n > \gamma / (|X| + \gamma)$.
    \end{itemize}
    Then, a feasible outcome $W \in \calF$ satisfies
    \emph{base extended justified representation (BEJR)}
    if for every $\gamma \in \mathbb{N}$
    and every group of voters $S \subseteq \calN$
    that deserves a satisfaction level $\gamma$
    there is a voter $i \in S$ for which
    $\sat_i(W) \ge \gamma$.
\end{definition}

Intuitively, we decide what level of satisfaction $S$ deserves
by looking at feasible sets $X$
that we can interpret as possible counterproposals
for a solution that $S$ would like.
In order to prove that it deserves a satisfaction of $\gamma$,
subset $S$ has to show that
each such counterproposal $X$
is either compatible with the demands of $S$ (condition i)
or is too large (condition ii).
To show the former,
$S$ has to find their own proposal $Y$ of size at least $\gamma$,
universally approved in $S$,
that can be combined with $X$ in an outcome,
i.e., $X \cup Y \in \calF$.
For the latter,
the proportion between the size of $|X|$
and the number of voters outside of $S$
(so all potential voters behind the counterproposal)
must be larger than
the ratio between $\gamma$ and the number of voters in $S$, i.e.,
$|X|/(n - |S|) > \gamma / |S|$.
This inequality is equivalent to 
$|S|/n > \gamma / (|X| + \gamma)$ that appears in the definition.


As we will show later in this subsection,
in the temporal voting setting
BEJR implies EJR.%
\footnote{This observation was made also by \citet{chandak2024proportional}.}
However, the converse is not true,
i.e., EJR is strictly weaker than BEJR.
To see this consider the following example.

\begin{example}
\label{example:ejr:droop:hare}
    Let $E = (C,N,\ell,A)$ be
    a temporal election with $\ell=1$ round,
    2 candidates $C=\{c_1,c_2\}$,
    3 voters $N=\{1,2,3\}$, and
    an approval profile in which
    voters $1$ and $2$ approve candidate $c_1$
    while voter $3$ approves $c_2$, i.e.,
    $a_{1,1} = a_{2,1} = \{c_1\}$ and $a_{3,1} = \{c_2\}$.
    Observe that the outcome $\outcome = (c_2)$ satisfies EJR
    (and even sCore and sEJR+).
    Indeed, voter $3$ has maximal possible satisfaction,
    and for the subset of voters $S=\{1,2\}$ we have
    $\floor{\ell \cdot |S|/n} = \floor{2/3} = 0$,
    hence they do not receive any satisfaction guarantee from EJR.
    
    However, according to BEJR,
    they do deserve satisfaction of $\gamma=1$.
    To see this, consider all possible counterproposals $X$, i.e.,
    all feasible subsets of candidates $\calC_E$, i.e.,
    $\varnothing, \{(c_1, 1)\}, \{(c_2,1)\} \in \calF_E$.
    For $X = \varnothing$ or $X = \{(c_1,1)\}$,
    we can clearly find a set $Y = \{(c_1,1)\}$
    that is compatible with $X$, i.e., $Y \cup X \in \calF_E$
    and we have that $Y \in A_1 \cap A_2$ as well as $|Y| = 1 = \gamma$.
    On the other hand, if $X = \{(c_2, 1)\}$,
    then observe that $|S|/n = 2/3 > 1/2 = \gamma/(|X| + \gamma)$.
    Hence, $S$ deserves a satisfaction of $1$,
    which is not met by $\outcome$.
    Thus, outcome $\outcome$ violates BEJR.
\end{example}

Informally, the difference between BEJR and EJR
can be explained as follows.
For a subset of voters $S \subseteq N$
that agree in certain $t$ rounds,
EJR gives a satisfaction guarantee of $\gamma$
if the fraction of voters that $S$ constitute
is greater or equal to the fraction that $\gamma$ is of $t$, i.e.,
$|S|/n \ge \gamma/t$ or, equivalently, 
\begin{equation}
\label{eq:ejr:droop:hare:comparison:1}
    \gamma \le t \cdot |S|/n.
\end{equation}
However, for BEJR it is enough that $S$ is large enough
so that no other group of voters
can make claim to $t - \gamma + 1$ rounds.
So, the fraction of voters that are inside $S$
has to be strictly larger than $\gamma/(t+1)$
(as then no other group would constitute a fraction
of at least $(t - \gamma + 1)/(t + 1)$).
This gives us inequality 
$|S|/n > \gamma/(t+1)$ or, equivalently,
\begin{equation}
\label{eq:ejr:droop:hare:comparison:2}
    \gamma < (t + 1) \cdot |S|/n.
\end{equation}
This is a significant difference
as in many cases $\gamma$ may fail
Inequality~\eqref{eq:ejr:droop:hare:comparison:1}
but at the same time satisfy
Inequality~\eqref{eq:ejr:droop:hare:comparison:2},
which, for example, is the case in \Cref{example:ejr:droop:hare}.

The above difference is very similar to the difference
between two formulations of 
\emph{Proportionality for Solid Coalitions (PSC)} axiom
in the setting of
multiwinner voting with ordinal preferences
(see e.g.,~\citet{aziz2020expanding}).
Roughly speaking, the axiom says that for every subset of voters $S$
that all agree that candidates in a certain subset $X$
are better than every candidate outside of this subset,
at least $\min(\eta, |X|)$ candidates
need to be select to the committee from $X$.
How big is $\eta$ depends on the formulation.
In \emph{Hare-PSC}, we set $\eta = \floor{k \cdot |S|/n}$,
where $k$ is the size of the committee
and $n$ the total number of voters,
which resembles the approach of EJR
and Inequality~\eqref{eq:ejr:droop:hare:comparison:1}.
In \emph{Droop-PSC}, we set $\eta = \ceil{(k + 1) \cdot |S|/n} - 1$, i.e.,
we take the largest integer strictly smaller than $(k + 1) \cdot |S|/n$,
which matches Inequality~\eqref{eq:ejr:droop:hare:comparison:2}.

Following this distinction,
we can define \emph{Droop-EJR} for temporal voting as follows.

\begin{definition}
    Given a temporal election
     $E =( C ,N , \ell, A )$
    an outcome $\outcome$
    provides \emph{Droop extended justified representation (Droop-EJR)}
    if for every $t>0$ and every subset of voters $S \subseteq N$
    that agree in at least $t$ rounds,
    there exists a voter $i \in S$ such that
    $\sat_{i}(\outcome) \geq \ceil{(t+1) \cdot |S|/n} - 1$.
\end{definition}

Observe that Droop-EJR implies EJR
as to every subset of voters $S \subseteq N$
it gives the same or larger satisfaction guarantee.
We will now show
that BEJR implies Droop-EJR,
which then means that BEJR implies EJR as well.
The converse statement,
i.e., that Droop-EJR implies BEJR,
is also true under an additional assumption
that there is no voter that in some round
approves all of the candidates.
We need this assumption
as BEJR can give even higher satisfaction guarantees to groups of voters
that universally approve all available candidates in some rounds.

\begin{proposition}
\label{prop:bejr:ejr:equivalence}
    For every temporal election $E = (C,N,\ell,A)$,
    every outcome $\outcome$ that provides BEJR
    provides Droop-EJR as well.
    If additionally for every voter $i \in N$
    and every round $r \in [\ell]$ it holds that
    $a_{i,r} \neq C$,
    then every outcome $\outcome$ that provides Droop-EJR
    provides BEJR as well.
\end{proposition}
\begin{proof}
    Let us start by showing that BEJR implies Droop-EJR.
Take an arbitrary temporal election $E = (C,N,\ell,A)$
and outcome $\outcome$ that does not provide Droop-EJR.
Thus, there is a set of voters $S \subseteq N$
that agree in a subset of rounds $R$, with $|R|=t$,
and for each $i \in S$ we have
$\sat_i(\outcome) < \ceil{(t+1) \cdot |S|/n} - 1$.
We will show that according to BEJR
$S$ deserves a satisfaction guarantee of at least
$\gamma = \ceil{(t+1) \cdot |S|/n} - 1$,
which will mean that $\outcome$ fails BEJR as well.

To this end, take an arbitrary subset $X \in \calF_E$
such that $|S|/n \le \gamma / (|X| + \gamma)$.
Let us denote the subset of rounds that appear in $X$
by $R_X = \{ r : (c,r) \in X\}$.
We want to show that there is a subset $Y \subseteq \calC_E$
universally approved by voters in $S$
such that $|Y| \ge \gamma$ and
$X \cup Y \in \calF_E$.
Note that the last condition is equivalent to saying
that $R_X$ and the subset of rounds that appear in $Y$
are disjoint.

Observe that $|S|/n \le \gamma / (|X| + \gamma)$ is equivalent to
$(|X|+\gamma)/\gamma \le n/|S|$, from which we get
$|X|/\gamma \le n/|S| - 1$.
Thus,
\begin{align*}
    |X| &\le \gamma \cdot \frac{n}{|S|} - \gamma \\
    \tag{as $\gamma = \ceil{(t+1) \cdot |S|/n} - 1$}
    &= \left(\left\lceil(t+1) \cdot \frac{|S|}{n}\right\rceil - 1 \right)\cdot\frac{n}{|S|} - \gamma\\
    \tag{since $\ceil{x}<x+1$ for every $x \in \mathbb{R}$}
    &< \left((t+1) \cdot \frac{|S|}{n}\right)\cdot\frac{n}{|S|} - \gamma\\
    &= t+1 -\gamma.
\end{align*}
However, as $|X|$ needs to be an integer,
this actually implies that $|X|\le t - \gamma$
and consequently that $|R_X \cap R| \le t - \gamma$.
Hence, there is subset of rounds $R_Y \subseteq R$ such that
$R_Y \cap R_X = \varnothing$ and $|R_Y| = \gamma$.
Since $S$ agrees in every round in $R$,
for each round $r \in R_Y \subseteq R$
we can find candidate $c_r$
such that $c_r \in \bigcap_{i \in S} a_{i,r}$.
Then, we set $Y = \{(c_r, r) : r \in R_Y\}$,
which meets all of the desired criteria.
This means that $S$ deserves a satisfaction guarantee of $\gamma$
according to BEJR,
which concludes the proof that BEJR implies Droop-EJR.

Now, let us show that Droop-EJR implies BEJR
when there is no voter that in some round approves all of the candidates.
To this end,
take an arbitrary temporal election $E = (C,N,\ell,A)$
such that $a_{i,r} \neq C$,
for every $i \in N$ and $r \in [\ell]$,
and an arbitrary outcome $\outcome$
that does not provide BEJR.
This means that there exists a subset of voters $S$
that deserves a satisfaction guarantee $\gamma$,
but for every voter $i \in S$
it holds that $\sat_i(\outcome) < \gamma$.
We will show that $S$ receives guarantee $\gamma$ by Droop-EJR as well.

Let $R \subseteq [\ell]$ be a subset of all rounds
in which voters in $S$ agree and let $t = |R|$.
We will show that necessarily
$\gamma \le \ceil{(t+1) \cdot |S|/n} - 1$,
which will prove the thesis.

First let us show that $t \ge \gamma$.
To this end, take $X = \varnothing \in \calF_E$.
According to the definition of BEJR,
since $S$ deserves satisfaction $\gamma$,
there must be a subset of candidates $Y \subseteq \calC_E$
universally approved in $S$, such that $|Y| \ge \gamma$ and
$Y = Y \cup \varnothing \in \calF_E$.
By definition, for every $(c,r) \in Y$
it must hold that $c \in a_{i,r}$,
for every $i \in S$.
Hence, voters in $S$ agree on at least $\gamma$ rounds.
Thus, indeed $t \ge \gamma$.

Now, let us take arbitrary subset of rounds
$R_X \subseteq R$ such that
$|R_X| = t + 1 - \gamma$.
For every round $r \in R_X$,
let $c_r \in C \setminus a_{i,r}$
be a candidate not approved by arbitrary voter $i \in S$
(by the assumption, we know that such exists).
Then, let $X = \{(c_r,r) : r \in R_X\}$.

Observe that now $|R \setminus R_X| < \gamma$.
This means that there is no subset of rounds $R_Y \subseteq [\ell]$
with $|R_Y| \ge \gamma$
that would be disjoint with $R_X$ and
in which all voters in $S$ agree.
And since for every $(c,r) \in X$
there is a voter in $S$
that does not approve candidate $c$ in round $r$,
there is no set $Y$ that could satisfy
condition (i) of \Cref{def:bejr}.
Hence, it must be the condition (ii) that holds, i.e.,
$|S|/n > \gamma / (|X| + \gamma) = \gamma/(t+1)$.
Observe that this is equivalent to
$\gamma < (t+1) \cdot |S|/n$.
Now, for every positive real $x$,
the largest integer value strictly smaller than $x$
is $\ceil{x}-1$.
Thus, it must hold that
$\gamma \le \ceil{(t+1) \cdot |S|/n} - 1$,
which concludes the proof of the thesis.
\end{proof}

\subsection{FJR}

\citet{MasPieSko-2023-GeneralProportionality} also introduced an axiom
that is an analogue of FJR in the general approval voting setting.

\begin{definition}
\label{def:bfjr}
    Given a general approval election instance 
    $I = (\calC,\calN,\calF,\calA)$
    a group of voters $S \subseteq \calN$
    is $(\alpha,\beta)$-\emph{BFJR-cohesive},
    if for each feasible set $X \in \calF$, it holds that either
    \begin{itemize}
        \item[i)] there exists $Y \subseteq \calC$ with $|Y| = \alpha$ such that $\sat_i(Y) \ge \beta$, for every $i \in S$, and $X \cup Y \in \calF$, or
        \item[ii)] $|S|/n > \alpha / (|X| + \alpha)$.
    \end{itemize}
    Then, a feasible outcome $W \in \calF$ satisfies
    \emph{base full justified representation (BFJR)}
    if for every $\alpha, \beta \in \mathbb{N}$
    and every $(\alpha, \beta)$-BFJR-cohesive group of voters $S \subseteq \calN$
    there is a voter $i \in S$ for which
    $\sat_i(W) \ge \beta$.
\end{definition}

Observe that BFJR is a stronger axiom than BEJR
as BEJR gives the same guarantees but only to $(\gamma, \gamma)$-BFJR-cohesive groups.

Analogously to what we have proved in \Cref{prop:bejr:ejr:equivalence},
we can show that BFJR implies
a stronger version of the FJR axiom,
which we can define as follows.

\begin{definition}
\label{def:droopFJR}
    Given a temporal election $E =( C ,N , \ell, A )$,
    we say that 
    an outcome $\outcome$ provides 
    \emph{Droop full justified representation} (Droop-FJR)
    if for every subset of voters $S \subseteq N$
    and every subset of rounds $T \subseteq [\ell]$
    there exists subset $R \subseteq T$
    such that $|R| = \ceil{(|T| + 1) \cdot |S|/n} - 1$
    and for every other outcome $\outcome'$,
    there is a voter $i \in S$ with
    \(
        \sat_i(\outcome) \ge \min_{j \in S} \sat_j(\outcome'_R).
    \)
    Equivalently, for every subset of voters $S$,
    there is $i \in S$ such that
    \[
        \sat_i(\outcome) \ge 
        \max_{T \subseteq [\ell]} \
        \min_{R \subseteq T:|R|= \ceil{(|T|+1) \cdot |S|/n} -1} \
        \max_{\outcome'_R \in C^R} \
        \min_{j\in S} \
        \sat_j(\outcome'_R).
    \]
\end{definition}

Note that Droop-FJR implies FJR
as for every subset of voters $S \subseteq N$
and every subset of rounds $T$,
Droop-FJR gives at least as many rounds $R$
to be used by the voters in $S$
as does the standard FJR.
Thus, the minimum satisfaction from
suboutcome $\outcome'_R$
cannot be smaller for Droop-FJR than for FJR.

As we will now show,
BFJR implies Droop-FJR in the temporal voting setting,
which means that it implies FJR as well.
We leave open the question whether,
possibly under some additional assumptions,
the converse implication is true, i.e.,
whether Droop-FJR implies BFJR.

\begin{proposition}
    For every temporal election $E = (C,N,\ell,A)$,
    every outcome $\outcome$ that provides BFJR
    provides Droop-FJR as well.
\end{proposition}
\begin{proof}
    Fix arbitrary temporal election $E = (C,N,\ell,A)$
    and $\outcome$ that does not provide Droop-FJR.
    This means that there exists a subset of voters $S \subseteq N$
    such that for every $i \in S$
    it holds that $\sat_i(\outcome) < \mu$,
    where
    \[
        \mu = 
        \max_{T \subseteq [\ell]} \
        \min_{R \subseteq T:|R|=\ceil{(|T|+1) \cdot |S|/n} -1} \
        \max_{\outcome'_R \in C^R} \
        \min_{j\in S} \
        \sat_j(\outcome'_R).
    \]
    Let us fix an arbitrary subset $T \subseteq [\ell]$
    that provides the above equality
    and denote $t = |T|$
    as well as $\alpha = \ceil{(t+1)\cdot |S|/n} - 1$.
    We will show that $S$ is $(\alpha,\mu)$-BFJR-cohesive,
    which will imply that $\outcome$ does not provide BFJR as well.

    Take an arbitrary subset $X \in \calF_E$
    such that $|S|/n \le \alpha / (|X| + \alpha)$.
    Let us denote the subset of rounds in $X$
    as $R_X = \{ r: (c,r) \in X\}$.
    Our goal is to show that there is a subset $Y \subseteq \calC_E$
    such that $|Y| = \alpha$, 
    $\sat_i(Y) \ge \beta$, for every $i \in S$, and
    $X \cup Y \in \calF_E$.
    The last condition simply says
    that the subset of rounds in $Y$
    is disjoint from $R_X$.

    Now, $|S|/n \le \alpha / (|X| + \alpha)$
    is equivalent to
    $(|X|+\alpha)/\alpha \le n/|S|$, which gives us
    $|X|/\alpha \le n/|S| - 1$.
    Further,
    \begin{align*}
        |X| &\le \alpha \cdot \frac{n}{|S|} - \alpha \\
        \tag{as $\alpha = \ceil{(t+1) \cdot |S|/n} - 1$}
        &= \left(\left\lceil(t+1) \cdot \frac{|S|}{n}\right\rceil - 1 \right)\cdot\frac{n}{|S|} - \alpha\\
        \tag{since $\ceil{x}<x+1$ for every $x \in \mathbb{R}$}
        &< \left((t+1) \cdot \frac{|S|}{n}\right)\cdot\frac{n}{|S|} - \alpha\\
        &= t+1 -\alpha.
    \end{align*}
But $|X|$ needs to be an integer,
which means that actually $|X| \le t - \alpha$.
Thus, $|R_X \cap T| \le t - \alpha$ as well,
which implies that there is a subset of rounds
$R_Y \subseteq T$ such that
$R_Y \cap R_X = \varnothing$ and $|R_Y| = \alpha$.
Then, observe that
\[
        \max_{\outcome'_{R_Y} \in C^{R_Y}} \
        \min_{j\in S} \
        \sat_j(\outcome'_R) \ge
        \min_{R \subseteq T:|R|=\ceil{(|T|+1) \cdot |S|/n} -1} \
        \max_{\outcome'_R \in C^R} \
        \min_{j\in S} \
        \sat_j(\outcome'_R) = \mu.
\]
In other words, there exists an outcome $\outcome'$
such that for every voter $i \in S$
its satisfaction from $\outcome'$
in rounds $R_Y$ is at least $\mu$.
We can thus take $Y = \{ (o'_r, r) : r \in R_Y\}$
and it will hold that $\sat_i(Y) \ge \mu$
for every $i \in S$.
Also, $|Y| = \alpha$ and, since $R_X$ and $R_Y$ are disjoint,
$X \cup Y \in \calF_E$.
Thus, indeed $S$ is $(\alpha,\mu)$-BFJR-cohesive,
which concludes the proof.
\end{proof}

\section{Omitted Proofs from Section \ref{sec:ejr-plus}}

\PropVsImplications*
\begin{proof}
Consider a subset of voters $S$ that agree in $t$ rounds.
We have $t\le \ell$, $|S|/n\le 1$ and hence $\lfloor t\cdot|S|/n\rfloor\le \min(t, \lfloor\ell\cdot|S|/n\rfloor)$. This shows that sEJR implies EJR and sPJR implies PJR:
in both cases, the strong version of the axiom provides a stronger guarantee to each group of voters. Similarly, sJR implies JR because $t\le \ell$.

We will now argue that sEJR+ implies sEJR.
Consider an outcome $\outcome$ that provides sEJR+, and a subset of voters $S$ who agree in a size-$t$ subset of rounds $R$. We need to show that the satisfaction of some voter $i\in S$ is at least
$\min(t, \lfloor\ell\cdot|S|/n\rfloor)$.
Suppose first that for each round $r\in R$ 
the outcome $o_r$ is approved by all voters in $S$. Then the satisfaction of each voter in $S$ is at least $|R|=t\ge \min(t, \lfloor\ell\cdot|S|/n\rfloor)$. Otherwise,
consider a round $r\in R$ such that not all voters in $S$ approve $o_r$. By the choice of $r$, there is some candidate $c$ that is approved by all voters in $S$ in round $r$.
Then $o_r\neq c$ and hence sEJR+ implies that 
$\sat_i(\outcome)\ge \lfloor\ell\cdot|S|/n\rfloor\ge \min(t, \lfloor\ell\cdot|S|/n\rfloor)$ for some $i\in S$. 

It remains to show that sEJR implies sPJR and sPJR implies sJR; this follows by the same arguments as in the multiwinner case.
\end{proof}

\EJRplusImplications*
\begin{proof}
    For part i),
    consider an arbitrary election $E = (C, N, \ell, A)$
    and an outcome $\outcome$ that does not provide EJR+.
    We will show that $\outcome$ does not provide sEJR+ either.
    By the definition of EJR+ there exist
    $\sigma \in [n]$, $\tau \in [\ell]$, 
    a $(\sigma, \tau)$-cohesive subset of voters $S \subseteq N$,
    a round $r \in [\ell]$, and a candidate $c \in C$ approved by all voters in $S$ in round $r$ such that 
    $o_r\neq c$ and for each voter $i \in S$ we have $\sat_i(\outcome) < \floor{\tau \cdot \sigma/n}$.
    Then, since $\tau \le \ell$ and $\sigma \le |S|$, 
    for each voter $i \in S$ it holds that 
    \[
        \sat_i(\outcome) < 
        \floor{\tau \cdot \sigma /n}
        \le \floor{\ell \cdot |S|/n}.
    \]
    Hence, $S$ is a witness that 
    $\outcome$ fails to provide 
    sEJR+.

    For part ii),
    consider an arbitrary election $E = (C, N, \ell, A)$
    and an outcome $\outcome$ that does not provide EJR.
    We will show that $\outcome$ does not provide EJR+ either.
    By the definition of EJR, there exists a 
    $\tau \in [\ell]$, a set of rounds $R$, $|R|=\tau$,
    and a subset of voters $S$ 
    such that all voters in $S$ agree in each round $r\in R$, but
    for every $i \in S$ we have
    $\sat_i(\outcome) < \floor{\tau \cdot |S|/n}$.
    As $|S|\le n$, this means that
    $\sat_i(\outcome) < \tau$,
    which implies that there is a round $r\in R$ such that some voter in $S$ does not approve $o_r$, 
    and hence $o_r\not\in\bigcap_{i\in S}a_{i, r}$.
    Let $\sigma = |S|$ and note that $S$ is $(\sigma, \tau)$-cohesive, as witnessed by $R$.
    Hence, $S$ witnesses that $\outcome$ fails to provide EJR+.
\end{proof}

\WEJRplusrelationship*

\begin{proof}
    We will first show that EJR+ implies wEJR+.
    Consider an arbitrary election $E = (C, N, \ell, A)$ and an outcome $\outcome$ that satisfies EJR+.
    For any $\sigma \in [n]$ and $\tau \in [\ell]$, for any $(\sigma, \tau)$-cohesive subset of voters $S \subseteq N$ it holds that either there exists a voter $i \in S$ such that $\sat_i(\outcome) \ge \floor{\tau \cdot \sigma/n}$, or for every round $r\in [\ell]$ with $\bigcap_{i\in S}a_{i, r}\neq\varnothing$ we have that $o_r \in \bigcap_{i \in S} a_{i,r}$.

    Consider the case when $\tau = \ell$.
    For those $(\sigma, \ell)$-cohesive subsets of voters $S$ where there exists a voter $i \in S$ such that $\sat_i(\outcome) \ge \floor{\tau \cdot \sigma/n} = \floor{\ell \cdot \sigma/n}$, the first condition of wEJR+ is satisfied.
    Otherwise, for those $(\sigma, \ell)$-cohesive subsets of voters $S$ where no such voter $i \in S$ exists, by definition of EJR+ we know that for every round with $\bigcap_{i\in S}a_{i, r}\neq\varnothing$ we have that $o_r \in \bigcap_{i \in S} a_{i,r}$.
    This satisfies the second condition of wEJR+, hence at least one of the two wEJR+ conditions always holds, so $\outcome$ provides wEJR+.

    It remains to show that wEJR+ implies wEJR.
    Consider an arbitrary election $E = (C, N, \ell, A)$ and an outcome $\outcome$ that does not provide wEJR.
    We will show that $\outcome$ does not provide wEJR+ either.
    By definition of wEJR, there exists a subset of voters $S$ such that all voters in $S$ agree in each round $r \in [\ell]$, but for every $i \in S$ we have $\sat_i(\outcome) < \floor{\ell \cdot |S|/n}$.
    As $|S| \leq n$, this means that $\sat_i(\outcome) < \ell$, which implies there is a round $r \in [\ell]$ such that some voter in $S$ does not approve $o_r$, and  hence $o_r \notin \bigcap_{i \in S} a_{i,r}$.
    Let $\sigma = |S|$ and note that $S$ is $(\sigma, \ell)$-cohesive,. Hence, $S$ witnesses that $\outcome$ fails to provide wEJR+.
\end{proof}

\section{Omitted Proofs from Section \ref{sec:fjr}}
\PropFJRimplications*
\begin{proof}
    For part i), we first show that sFJR implies sEJR. Consider an arbitrary election $E = (C, N, \ell, A)$ and an outcome $\outcome$ that provides sFJR. Let $S \subseteq N$ be a subset of voters that agree in a set of $t$ rounds $T$, $t>0$, in election $E$ (note that if $t=0$ then $S$ cannot be a witness that sEJR is violated).

    First, suppose that $t \geq \floor{\ell \cdot |S|/n}$. Let $T' \subseteq T$ be a subset of $T$ of size $|T'| = \floor{\ell \cdot |S|/n}$. Then voters in $S$ agree in every round $r\in T'$, so there exists an outcome $\outcome'$ such that for all voters $i \in S$ it holds that $\sat_i(\outcome'_{T'}) = \floor{\ell \cdot |S|/n}$. Since $\outcome$ provides sFJR, there exists a voter $i \in S$ such that
    \[
        \sat_i(\outcome) \ge 
        \max_{R \subseteq [\ell]:|R|= \floor{\ell \cdot |S|/n}} \
        \max_{\outcome''_R \in C^R} \
        \min_{j\in S} \
        \sat_j(\outcome''_R) \ge
    \min_{j\in S} \
        \sat_j(\outcome'_{T'}) =
    \floor{\ell \cdot |S|/n} = \min(t, \floor{\ell \cdot |S|/n})
    \]
    Hence, sEJR is satisfied.

    On the other hand, suppose that $t < \floor{\ell \cdot |S|/n}$. We extend $T$ to a subset of rounds $T' \subseteq [\ell]$ such that $|T'| = \floor{\ell \cdot |S|/n}$ and $T\subseteq T'$. Consider an outcome $\outcome'$ that satisfies $o'_r\in\bigcap_{i\in S}a_{i, r}$ for all $r\in T$.
    Then for each $i\in S$ we have $\sat_i(\outcome'_{T'}) \ge t$. Given that $\outcome$ provides sFJR, by definition there exists a voter $i \in S$ such that 
    $$
    \sat_i(\outcome) \ge 
        \max_{R \subseteq [\ell]:|R|= \floor{\ell \cdot |S|/n}} \
        \max_{\outcome''_R \in C^R} \
        \min_{j\in S} \
        \sat_j(\outcome''_R) \ge
    \min_{j\in S} \
        \sat_j(\outcome'_{T'}) \ge
     t = \min(t, \floor{\ell \cdot |S|/n}).
     $$
    Hence, sEJR is satisfied in this case as well.

    We will now show that sFJR implies FJR. Fix a temporal election $E$, and consider an outcome $\outcome$ that does not provide FJR.
    Then there exists a subset of voters $S \subseteq N$ such that 
    for all $i\in S$ it holds that $\sat_i(\outcome)<\max_{T\subseteq [\ell]}\mu_S(T)$. This means 
    that for some subset of rounds $T \subseteq [\ell]$ it holds that for each subset $R \subseteq T$ of size $\floor{|T| \cdot |S|/n}$ there exists an outcome $\outcome'$ such that \(
        \sat_i(\outcome) < \min_{j \in S} \sat_j(\outcome'_R)
    \) for all voters $i \in S$.
    Let $R \subseteq T$ be some such subset of rounds. Extend it to a subset of rounds $R' \subseteq [\ell]$ such that $R \subseteq R'$ and $|R'| = \floor{\ell \cdot |S|/n}$. Then for $R'$ and for all voters $i \in S$ we have
    $$
    \sat_i(\outcome) < \min_{j \in S} \sat_j(\outcome'_R) \leq \min_{j \in S} \sat_j(\outcome'_{R'}).
    $$
    Hence, $S$, $R'$ and $\outcome'$ witness that $\outcome$ does not satisfy sFJR. 

    For part~ii), we first show that FJR implies EJR. Take an arbitrary election $E = (C, N, \ell, A)$ and an outcome $\outcome$ that provides FJR. Let $S \subseteq N$ be a subset of voters that agree in a set of rounds $T$ of size $|T|=t$ in election $E$. Consider an outcome $\outcome'$ that satisfies $o'_r\in\bigcap_{i\in S}a_{i, r}$ for each $r\in T$. 
    Observe that for every subset of rounds $R \subseteq T$ with $|R| = \floor{t \cdot |S|/n}$ we have $\sat_j(\outcome'_R)=\floor{t \cdot |S|/n}$ for each $j\in S$.
    Consequently, 
$$
\mu_S(T) = \min_{R \subseteq T:|R|= \floor{t \cdot |S|/n}} \
        \max_{\outcome''_R \in C^R} \
        \min_{j\in S} \
        \sat_j(\outcome''_R) \ge 
        \min_{R \subseteq T:|R|= \floor{t \cdot |S|/n}} \
        \min_{j\in S} \
        \sat_j(\outcome'_R) \ge
        \floor{t \cdot |S|/n}.
$$
    Since $\outcome$ provides FJR, there exists a voter $i \in S$ such that 
    $$
    \sat_i(\outcome) \geq 
    \max_{T'\subseteq [\ell]}
\mu_S(T') \ge
     \mu_S(T) \geq \floor{t \cdot |S|/n}.
    $$ 
    Hence, $\outcome$ provides EJR.

    We will now argue that FJR implies wFJR.
    Consider an outcome $\outcome$ that provides FJR for a temporal election $E$. Then for every subset of voters $S$ we have $\sat_i(\outcome)\ge \max_{T\subseteq [\ell]}\mu_S(T)$. Thus, in particular, 
    for $T=[\ell]$ we have 
    $$
    \sat_i(\outcome)\ge \mu_S([\ell]) = 
    \min_{R \subseteq [\ell]:|R|= \floor{\ell \cdot |S|/n}} \
        \max_{\outcome'_R \in C^R} \
        \min_{j\in S} \
        \sat_j(\outcome'_R).
    $$
    This means that $\outcome$ provides wFJR.

    For part iii), we will show that wFJR implies wEJR. Fix an election $E = (C, N, \ell, A)$ and an outcome $\outcome$ that provides wFJR. Suppose there exists a subset of voters $S \subseteq N$ such that $\bigcap_{i\in S}a_{i, r}\neq\varnothing$ for all $r\in [\ell]$.
    We construct an outcome $\outcome'$ so that $o'_r\in\bigcap_{i\in S}a_{i, r}$ for each $r\in [\ell]$.
    Then, for each subset of rounds $R \subseteq [\ell]$ such that $|R| = \floor{\ell \cdot |S|/n}$ 
    we have $\sat_j(\outcome'_R)=\floor{\ell \cdot |S|/n}$ for all $j\in S$.
    Since $\outcome$ provides wFJR, there exists some voter $i \in S$ such that 
    $$
    \sat_i(\outcome)\ge 
    \min_{R \subseteq [\ell]:|R|= \floor{\ell \cdot |S|/n}} \
        \max_{\outcome''_R \in C^R} \
        \min_{j\in S} \
        \sat_j(\outcome''_R) \ge  
    \min_{R \subseteq [\ell]:|R|= \floor{\ell \cdot |S|/n}} \
        \min_{j\in S} \
        \sat_j(\outcome'_R) \ge  \floor{\ell \cdot |S|/n}.
    $$
 Hence, $\outcome$ provides wEJR.
\end{proof}

\section{Omitted Proofs from Section \ref{sec:fpjr}}
\label{app:fpjr}

We first state and prove the result that was implied by footnote \ref{footnote_fpjr}.

\begin{proposition}
\label{prop:sejr:not-satisfiable}
    There exists an $n$-voter, $\ell$-round temporal election $E \in \calE_{=1}$
    such that $n | \ell$ and
    no outcome $\outcome$ provides sEJR for $E$.
\end{proposition}
\begin{proof}
    Consider an election $E=(C,N,\ell,A)$ with $\ell = 6$ rounds, a set of voters $N=[6]$, and a set of candidates $C = \{x, y, z, c_1,\dots,c_{6}\}$.
The voters' ballots are as follows (each $a_{i, r}$ is a singleton, so we omit the braces):

\begin{table}[ht]
    \setlength{\tabcolsep}{3pt}
    \renewcommand{\arraystretch}{1.2}
    \small
    \begin{center}
        \begin{tabular}{c | cc cc cc }
        \toprule
        & \multicolumn{6}{c}{Voter} \\
        Rounds & 1 & 2 & 3 & 4 & 5 & 6 \\
        \midrule
        1 & $x$ & $x$ & $y$ & $y$ & $z$ & $z$ \\
        2 & $x$ & $x$ & $y$ & $y$ & $z$ & $z$ \\
        3 & $c_1$ & $c_2$ & $c_3$ & $c_4$ & $c_5$ & $c_6$ \\
        4 & $c_1$ & $c_2$ & $c_3$ & $c_4$ & $c_5$ & $c_6$ \\
        5 & $c_1$ & $c_2$ & $c_3$ & $c_4$ & $c_5$ & $c_6$ \\
        6 & $c_1$ & $c_2$ & $c_3$ & $c_4$ & $c_5$ & $c_6$ \\
        \bottomrule
        \end{tabular}
    \end{center}
    \label{table:sEJR:not-satisfiable}
\end{table}
Intuitively, in the first two rounds the voters form three cohesive blocks of size two, and in the next four rounds, every voter supports a different candidate.

Observe that sEJR requires that every voter in $N$
should receive satisfaction of at least $1$.
Additionally, since each pair of voters $(1,2)$, $(3,4)$, and $(5,6)$ agrees in two rounds,
at least one voter from each of these pairs
needs to receive satisfaction of at least $2$,
according to sEJR.
We will show that it is not possible to satisfy both of this requirements at the same time.

Take an arbitrary outcome $\outcome$.
It can either select the same candidate in the first two rounds, or two different ones.
Let us consider both cases.

Case 1, $o_1 = o_2$.
Without loss of generality assume that
$o_1 = o_2 = x$.
Then, in order to guarantee satisfaction of at least $1$ 
to each of the voters $3$, $4$, $5$, and $6$,
we need to select a different candidate in each of the rounds $3$, $4$, $5$, and $6$.
However, in this way no voter in the pair $(3,4)$
will receive satisfaction of at least $2$.
Thus, sEJR cannot be satisfied if $o_1 = o_2$.

Case 2, $o_1 \neq o_2$.
Without loss of generality assume that
$o_1 = x$ and $o_2 = y$.
In order to guarantee that one voter in each of the pairs $(1,2)$ and $(3,4)$ has satisfaction of at least $2$,
in one round from $3$, $4$, $5$, or $6$,
we need to select candidate $c_1$ or $c_2$
and in one $c_3$ or $c_4$.
Then, in the remaining two rounds
we have to select candidates $c_5$ and $c_6$,
once each,
so to guarantee that both voter $5$ and voter $6$
have satisfaction of at least  $1$.
However, in this way no agent from the pair $(5,6)$
will have satisfaction of at least $2$.
Therefore, it is not possible to satisfy sEJR also in the case when $o_1 \neq o_2$,
which concludes the proof.
\end{proof}

\PropFPJRimplications*
\begin{proof}
For part i), we first show that sFJR implies sFPJR.
Take an arbitrary election $E = (C,N,\ell,A)$ and an outcome $\mathbf{o}$ satisfying sFJR.
For any subset of voters $S \subseteq N$, there exists an $i \in S$ such that
\begin{equation*}
    \sat_i(\outcome) \ge 
        \max_{R \subseteq [\ell]:|R|= \floor{\ell \cdot |S|/n}} \
        \max_{\outcome'_R \in C^R} \
        \min_{j\in S} \
        \sat_j(\outcome'_R).
\end{equation*}
Then, since $i \in S$, we must have that $\sat_S(\outcome) \geq  \sat_i(\outcome)$ and $\mathbf{o}$ satisfies sFPJR as well.

Next, we show that sFPJR implies sPJR.
Take an arbitrary election $E = (C,N,\ell,A)$ and an outcome $\mathbf{o}$ satisfying sFPJR.
Let $S \subseteq N$ be a subset of voters that agree in a set of $t$ rounds $T, t > 0$, in election $E$ (note that if $t = 0$ then $S$ cannot be a witness that sPJR violated).

First, suppose that $t \geq \floor{\ell \cdot |S|/n}$.
Let $T' \subseteq T$ be a subset of $T$ of size $|T'| = \floor{\ell \cdot |S|/n}$.
Then voters in $S$ agree in every round $r \in T'$, so there exists an outcome $\outcome'$ such that for all voters $i \in S$, it holds that $\sat_i(\outcome'_{T'}) = \floor{\ell \cdot |S|/n}$.
Since $\outcome'$ provides sFPJR, by definition, we have that
\begin{equation*}
    \sat_{S}(\outcome) \geq \max_{R \subseteq [\ell]:|R|= \floor{\ell \cdot |S|/n}} \
        \max_{\outcome''_R \in C^R} \
        \min_{j\in S} \
        \sat_j(\outcome''_R) \geq \min_{j\in S} \
        \sat_j(\outcome'_{T'}) \geq \floor{\ell \cdot |S|/n} = \min(t,\floor{\ell \cdot |S|/n}).
\end{equation*}
Hence, sPJR is satisfied.

On the other hand, suppose that $t < \floor{\ell \cdot |S|/n}$.
We extend $T$ to a subset of rounds $T' \subseteq [\ell]$ such that $|T'| = \floor{\ell \cdot |S|/n}$ and $T \subseteq T'$.
Consider an outcome $\outcome'$ that satisfies $o'_r \in \bigcap_{i \in S} a_{i,r}$ for all $r \in T$.
Then for each $i \in S$, we have $\sat_i(\outcome'_{T'}) \geq t$.
Given that $\outcome$ provides sFPJR, by definition, we have that
\begin{equation*}
    \sat_{S}(\outcome) \geq \max_{R \subseteq [\ell]:|R|= \floor{\ell \cdot |S|/n}} \
        \max_{\outcome''_R \in C^R} \
        \min_{j\in S} \
        \sat_j(\outcome''_R) \geq \min_{j\in S} \
        \sat_j(\outcome'_{T'}) \geq t = \min(t,\floor{\ell \cdot |S|/n}).
\end{equation*}
Hence, sPJR is satisfied in this case as well.

We now show that sFPJR implies FPJR.
Fix a temporal election $E = (C,N,\ell,A)$, and consider an outcome $\outcome$ that does not satisfy FPJR.
Then there exists a subset of voters $S \subseteq N$ and a subset of rounds $T \subseteq [\ell]$ such that for all subsets $R \subseteq T \subseteq [\ell]$ where $|R| = \floor{|T| \cdot |S|/n}$, there exists an outcome $\outcome'$ such that $\sat_S(\outcome) < \min_{i \in S} \sat_i(\outcome'_R)$.
Take one such round subset $R \subseteq T$ that violates the FPJR condition.
The round subset $R' \subseteq [\ell]$ such that $R \subseteq R'$ and $|R'| = \floor{\ell \cdot |S|/n}$ also violates the FPJR condition because 
\begin{equation*}
    \sat_S(\outcome) < \min_{i \in S} \sat_i(\outcome'_R) \leq \min_{i \in S} \sat_i(\outcome'_{R'}).
\end{equation*}
Note that $R'$ also violates the sFPJR condition, and hence $\outcome$ does not satisfy sFPJR.
As we assumed that $\outcome$ violates FPJR, the desired result follows by contraposition.

For part ii), we first show that FJR implies FPJR.
Take an arbitrary election $E = (C, N, \ell, A)$ and an outcome $\outcome$ satisfying FJR.
For any subset of voters $S \subseteq N$, there exists an $i \in S$ such that 
\begin{equation*}
    \sat_i(\outcome) \ge 
        \max_{T \subseteq [\ell]} \
        \min_{R \subseteq T:|R|= \floor{|T| \cdot |S|/n}} \
        \max_{\outcome'_R \in C^R} \
        \min_{j\in S} \
        \sat_j(\outcome'_R).
\end{equation*}
Then, since $i \in S$, we must have that $\sat_S(\outcome) \geq \sat_i(\outcome)$ and $\outcome$ satisfies FPJR as well.

Next, we show that FPJR implies PJR.
Take an arbitrary election $E = (C, N, \ell, A)$ and an outcome $\outcome$ that provides FPJR.
Let $S \subseteq N$ be a subset of voters that agree in a set of rounds $T$ of size $|T| = t$ in election $E$.
Consider an outcome $\outcome'$ that satisfies $o'_r \in \bigcap_{i \in S} a_{i,r}$ for each $r \in T$.
Observe that for every subset of rounds $R \subseteq T$ with $|R| = \floor{t \cdot |S|/n}$ we have $\sat_j(\outcome'_R) = \floor{t \cdot |S|/n}$ for each $j \in S$.
Consequently,
\begin{equation*}
    \mu_S(T) =
            \min_{R \subseteq T:|R|=\floor{t \cdot |S|/n}} \
            \max_{\outcome''_R \in C^R} \
            \min_{j \in S} \
            \sat_j(\outcome''_R) \geq \min_{R \subseteq T:|R|=\floor{t \cdot |S|/n}} \
            \min_{j \in S} \
            \sat_j(\outcome'_R) \geq \floor{t \cdot |S|/n}
\end{equation*}
Since $\outcome$ provides FPJR, we have that
\begin{equation*}
    \sat_S(\outcome) \geq 
    \max_{T' \subseteq [\ell]} \
            \mu_S(T') \geq \mu_S(T) \geq \floor{t \cdot |S|/n}.
\end{equation*}
Hence, $\outcome$ provides PJR.

We will now argue that FPJR implies wFPJR.
Consider an outcome $\outcome$ that provides FPJR for a temporal election $E$.
Then for every subset of voters $S$ we have $\sat_S(\outcome) \geq \max_{T \subseteq [\ell]} \mu_S(T)$.
Thus, in particular, for $T = [\ell]$, we have
\begin{equation*}
    \sat_S(\outcome) \geq \mu_S([\ell]) = \min_{R \subseteq [\ell]:|R|=\floor{\ell \cdot |S|/n}} \
            \max_{\outcome'_R \in C^R} \
            \min_{j\in S} \
            \sat_j(\outcome'_R).
\end{equation*}
This means that $\outcome$ provides wFPJR.

For part iii), we first show that wFJR implies wFPJR.
Take an arbitrary election $E = (C, N, \ell, A)$ and an outcome $\outcome$ satisfying wFJR.
For any subset of voters $S \subseteq N$, there exists an $i \in S$ such that 
\begin{equation*}
    \sat_i(\outcome) \ge 
        \min_{R \subseteq [\ell]:|R|= \floor{\ell \cdot |S|/n}} \
        \max_{\outcome'_R \in C^R} \
        \min_{j\in S} \
        \sat_j(\outcome'_R).
\end{equation*}
Then, since $i \in S$, we must have that $\sat_S(\outcome) \geq \sat_i(\outcome)$ and $\outcome$ satisfies wFPJR as well.

Next, we show that wFPJR implies wPJR.
Fix an election $E = (C, N, \ell, A)$ and an outcome $\outcome$ that provides wFPJR.
Suppose there exists a subset of voters $S \subseteq N$ such that $\bigcap_{i \in S} a_{i,r} \neq \varnothing$ for all $r \in [\ell]$.
We construct an outcome $\outcome'$ so that $o'_r \in \bigcap_{i \in S} a_{i,r}$ for each $r \in [\ell]$.
Then, for each subset of rounds $R \subseteq [\ell]$ such that $|R| = \floor{\ell \cdot |S|/n}$, we have $\sat_j(\outcome'_R) = \floor{\ell \cdot |S|/n}$ for all $j \in S$.
Since $\outcome$ provides wFPJR, we have that
\begin{equation*}
    \sat_S(\outcome) \geq \min_{R \subseteq [\ell]:|R|= \floor{\ell \cdot |S|/n}} \
        \max_{\outcome''_R \in C^R} \
        \min_{j\in S} \
        \sat_j(\outcome''_R)
        \geq \min_{R \subseteq [\ell]:|R|= \floor{\ell \cdot |S|/n}} \
        \min_{j\in S} \
        \sat_j(\outcome'_R) \geq \floor{\ell \cdot |S|/n}.
\end{equation*}
Hence, $\outcome$ provides wPJR.
\end{proof}

\PropwPJRiffwEJR*
\begin{proof}
It is trivially observable that wEJR implies wPJR, since for every subset of voters $S \subseteq N$ that agrees in all rounds and any outcome $\outcome$ satisfying wEJR,
    \begin{equation*}
        \sat_S(\outcome) \geq \sat_i(\outcome) \geq \left\lfloor\ell \cdot \frac{|S|}{n} \right\rfloor,
    \end{equation*}
    where the first inequality follows from the fact that $i \in S$ and the second inequality follows from the definition of wEJR.

    We now prove that wPJR implies wEJR for any election $E \in \calE_{=1}$.
    Note that any subset of voters $S\subseteq N$ that agree in all $\ell$ rounds must necessarily have the same ballot.
    This means that for any outcome $\outcome$, $\sat_S(\outcome) = \sat_i(\outcome)$, giving us the desired wEJR guarantee
    \begin{equation*}
        \sat_i(\outcome) = \sat_S(\outcome) \geq \left\lfloor \ell \cdot \frac{|S|}{n} \right\rfloor. \qedhere
    \end{equation*}
\end{proof}

\section{Omitted Proofs from Section \ref{sec:core}}
\PropCoreImplications*
\begin{proof}
    For part i), we first show that sCore implies sFJR. Take an arbitrary election $E = (C, N, \ell, A)$ and an outcome $\outcome$ satisfying sCore. For any subset of voters $S \subseteq N$, any subset of rounds $R \subseteq [\ell]$ such that $|R| = \floor{\ell \cdot |S|/n}$ and any outcome $\textbf{o}'$, it follows that $\sat_i(\outcome) \geq \sat_i(\outcome'_R)$ for some voter $i \in S$. Given that $\sat_i(\outcome) \geq \min_{j\in S} \
    \sat_j(\outcome)$ for every voter $i \in S$, $\outcome$ must also satisfy sFJR since the amount of satisfaction that sCore provides to a voter is at least as much as sFJR provides by definition.

    We now show that sCore implies Core by considering an outcome $\outcome$ that is not core stable. There then exists a subset of voters $S \subseteq N$ and a subset of rounds $T \subseteq [\ell]$ such that for any subset $R \subseteq T$ where $|R| = \floor{|T| \cdot |S|/n}$, there exists an outcome $\outcome'$ such that \(
        \sat_i(\outcome) < \sat_i(\outcome'_R)
    \) for all voters $i \in S$. Fix one such subset of rounds $R \subseteq T$ such that $\outcome'_R$ witnesses a violation of core stability, and consider a subset of rounds $R' \subseteq [\ell]$ such that $|R'| = \floor{\ell \cdot |S|/n}$ and $R \subseteq R'$ (i.e., $R'$ is an extension of $R$). Given that $\sat_i(\outcome'_R) \leq \sat_i(\outcome'_{R'})$ for all voters $i \in S$ and $\sat_i(\outcome) < \sat_i(\outcome'_R)$ since $\outcome$ is not core stable, it follows that $\sat_i(\outcome) < \sat_i(\outcome'_{R'})$ and so $\outcome$ is not strongly core stable either. As our original assumption was that $\outcome$ violated core stability, it follows by contraposition that $\outcome$ satisfying sCore implies that $\outcome$ satisfies Core.

    For part ii), consider an arbitrary election $E = (C, N, \ell, A)$ and an outcome $\outcome$ satisfying Core. For any subset of voters $S \subseteq N$ and any subset of rounds $T \subseteq [\ell]$, there therefore exists a subset $R \subseteq T$ such that $|R| = \floor{|T| \cdot |S|/n}$ and for any outcome $\textbf{o}'$, there exists some voter $i \in S$ such that $sat_i(\outcome) \geq \sat_i(\outcome'_R)$.
    
    We first show that Core implies FJR. As with the proof of sCore implying sFJR, we observe that Core guarantees at least as much satisfaction as FJR given that $\sat_i(\outcome) \geq \min_{j\in S} \
    \sat_j(\outcome)$ for every voter $i \in S$, therefore $\outcome$ automatically satisfies FJR too.

    We now show that Core implies wCore. If we look at the satisfaction guarantee made by $\outcome$ when $T = [\ell]$, we find that for any subset of voters $S \subseteq N$, there exists a subset of rounds $R \subseteq T = [\ell]$ such that $|R| = \floor{\ell \cdot |S|/n}$ and for any outcome $\textbf{o}'$, there exists some voter $i \in S$ such that $sat_i(\outcome) \geq \sat_i(\outcome'_R)$. Notice that this is exactly the definition of wCore, hence $\outcome$ satisfies wCore.

    For part iii), we show that wCore implies wFJR. Take an arbitrary election $E = (C, N, \ell, A)$ and an outcome $\outcome$ satisfying wCore. For any subset of voters $S \subseteq N$, there exists a subset of rounds $R \subseteq [\ell]$ with $|R| = \floor{\ell \cdot |S|/n}$ such that for any outcome $\outcome'$, there exists a voter $i \in S$ such that $\sat_i(\outcome) \geq \sat_i(\outcome'_R)$. Observe that for any voter $i \in S$ we have that $\sat_i(\outcome'_R) \geq \min_{j\in S} \
    \sat_j(\outcome'_R)$. It therefore follows that there exists a voter $i \in S$ such that $\sat_i(\outcome) \geq \min_{j\in S} \
    \sat_j(\outcome'_R)$, where $R$ is the same subset of rounds that enabled $\outcome$ to satisfy wCore, hence $\outcome$ also satisfies wFJR.
\end{proof}

\section{Omitted Proofs from Section \ref{sec:independence}}
As mentioned in the discussion in 
\Cref{sec:independence},
to rule out all remaining potential implications, we identify twelve specific pairs of axioms for which we construct outcomes that satisfy one notion but not the other.  
Therefore, to prove \Cref{prop:separation_prop}, it is equivalent to prove the following (which completes the `if' part; the `only if' part follows from all our results in previous sections).
\begin{restatable}{proposition}{PropIndependence}\label{prop:independence}
There exist elections
$E_1$,$E_2$,$E_3$,$E_4$,$E_5$,$E_6$,
$E_7$,$E_8$,$E_9 \in \calE_{=1}$ and $E_0 \in \calE_{\geq 1}$,
    as well as outcomes
    $\outcome_1$,\,$\outcome_2$,\,$\outcome_3$, $\outcome_4$,$\outcome_5$,$\outcome_6$,$\outcome_7$,$\outcome_8$,$\outcome_9$, and
    $\outcome_0$
    for respective elections, such that
\begin{itemize}
        \item[i)] in $E_1$, outcome $\outcome_1$ provides wCore \& wEJR+ but not JR,
        \item[ii)] in $E_2$, outcome $\outcome_2$ provides Core \& EJR+ but not sJR,
        \item[iii)] in $E_3$, outcome $\outcome_3$ provides sJR but not wPJR,
        \item[iv)] in $E_4$, outcome $\outcome_4$ provides sEJR+ but not wFPJR,
        \item[v)] in $E_5$, outcome $\outcome_5$ provides sFJR but not wCore,
        \item[vi)] in $E_6$, outcome $\outcome_6$ provides sCore but not wEJR+,
        \item[vii)] in $E_7$, outcome $\outcome_7$ provides sFPJR but not wFJR,
        \item[viii)] in $E_8$, outcome $\outcome_8$ provides sFPJR but not wEJR+,
        \item[ix)] in $E_9$, outcome $\outcome_9$ provides sFPJR but not EJR, and 
        \item[x)] in $E_0$, outcome $\outcome_0$ provides sFPJR but not wEJR.
    \end{itemize}
\end{restatable}

\begin{proof}
Let us consider each claim separately.
For each election in $\calE_{=1}$
we will omit the curly braces when
writing the ballots of the voters.
In the proof of Claim ii, we explain how any outcome of $E_2$ may be selected as $\outcome_2$. In the proofs of the other nine claims, we show that a specific outcome $\outcome$ witnesses the desired result, where the underlined candidate in row $r$ of the corresponding table indicates  the candidate that was selected in round $r$ of $\outcome$.

\vspace{0.2cm}
\noindent
\textbf{Claim i: wCore and wEJR+ do not imply JR even for $\calE_{=1}$}
\vspace{0.02cm}

\noindent
Let $E_1=(C,N,\ell,A) \in \calE_{=1}$ be a temporal election with
4 candidates $C = \{a,b, c, d\}$,
6 voters $N = [6]$,
$\ell = 4$ rounds,
and the ballots of voters as presented in \Cref{tab:prop:independence:i}.

\begin{table}[ht]
    \setlength{\tabcolsep}{3pt}
    \renewcommand{\arraystretch}{1.2}
    \small
    \begin{center}
        \begin{tabular}{c | ccc ccc }
        \toprule
        & \multicolumn{4}{c}{Voter} \\
        Rounds & 1 & 2 & 3 & 4 & 5 & 6 \\
        \midrule
        1 & $\underline{a}$ & $\underline{a}$ & $\underline{a}$ & $b$ & $b$ & $b$ \\
        2 & $\underline{a}$ & $\underline{a}$ & $\underline{a}$ & $b$ & $b$ & $b$ \\
        3 & $\underline{a}$ & $\underline{a}$ & $\underline{a}$ & $b$ & $c$ & $d$ \\
        4 & $\underline{a}$ & $\underline{a}$ & $\underline{a}$ & $b$ & $c$ & $d$ \\
        \bottomrule
        \end{tabular}
    \end{center}
    \caption{The ballots of voters in the proof of \Cref{prop:independence}, Claim i.}
    \label{tab:prop:independence:i}
\end{table}

Consider outcome $\outcome_1 = (a, a, a, a)$. First, let us show that $\outcome_1$ provides wCore. Given that voters 1, 2 and 3 approve the winner in every round of $\outcome_1$, none of these voters can be more satisfied by any other outcome of $E_1$, so any subset of voters containing one of these three voters cannot witness a violation of wCore. It therefore remains to check that no subset of voters $S \subseteq \{4, 5, 6\}$ can violate wCore.

We now consider the three possible cases for the size of such subsets of voters $S$. Let $\outcome'$ be an arbitrary outcome that witnesses a violation of wCore. If $|S| = 1$ then we find that $\floor{\ell \cdot |S|/n} = \floor{4 \cdot 1/6} = 0$, therefore any suboutcome $\outcome'_R$ causing the violation would have $|R| = \floor{\ell \cdot |S|/n} = 0$ rounds. Clearly no voter can be more satisfied by a suboutcome of 0 rounds than by $\outcome_1$, so wCore is not violated in this case. 

If $|S| = 2$ then we calculate that $|R| = \floor{\ell \cdot |S|/n} = 1$. However, taking $R$ to be round 3, we find that there exists a round subset $R$ where at least one of the two voters of $S$ must have a satisfaction of 0 given that no two voters approve a common candidate. This voter is therefore not more satisfied by any proposed suboutcome $\outcome'_R$ than by $\outcome_1$, so wCore is not violated.

If $|S| = 3$ then we calculate that $|R| = \floor{\ell \cdot |S|/n} = 2$. However, taking $R$ to be rounds 3 and 4, we find that there exists a round subset $R$ where at least one of the three voters of $S$ must have a satisfaction of 0 given that no two voters approve a common candidate in either round. This voter is therefore not more satisfied by any proposed suboutcome $\outcome'_R$ than by $\outcome_1$, so again wCore is not violated. This exhausts the possibilities of $|S|$, therefore the assumption that $\outcome'$ witnesses a violation of wCore is contradicted. As $\outcome'$ was an arbitrary outcome, we conclude that $\outcome_1$ satisfies wCore.

We now show that $\outcome_1$ provides wEJR+. First, we observe that voters 1, 2 and 3 receive utility $\ell$ from $\outcome_1$, therefore for any subset of voters $S \subseteq N$ containing any of these three voters, it holds that $\sat_i(\outcome) = \ell \ge \floor{\ell \cdot \sigma/n}$ for some $i\in S$ because $\sigma \leq n$. This satisfies the first condition of wEJR+, hence if a subset of voters exists that witnesses a violation of wEJR+, then this subset must itself be a subset of $\{4, 5, 6\}$.

Given that no two voters from the subset $\{4, 5, 6\}$ agree on a common candidate in rounds 3 or 4, it just remains to consider the three $(1, \ell)$-cohesive singleton subsets of $\{4, 5, 6\}$. Let $S'$ be one such singleton subset. By noting that $\floor{\ell \cdot \sigma/n} = \floor{4 \cdot 1/6} = 0$, it vacuously holds that $\sat_i(\outcome_1) \ge \floor{\ell \cdot \sigma/n}$ for the singleton voter $i\in S'$, hence the first condition of wEJR+ is satisfied. We therefore conclude that $\outcome_1$ satisfies wEJR+.

It remains to show that $\outcome_1$ violates JR. To see this, consider the subset of voters $S = \{4, 5, 6\}$, which agree on candidate $b$ in rounds 1 and 2 (so $t = 2$).
However, we find that
$$
\min(1, \floor{t \cdot |S|/n}) = \min(1, \floor{2 \cdot 3/6}) = \min(1, 1) = 1 > 0 = \sat_{S}(\outcome).
$$

Therefore, $S$ witnesses a violation of JR, hence $\outcome_1$ violates JR.

\vspace{0.2cm}
\noindent
\textbf{Claim ii: Core and EJR+ do not imply sJR even for $\calE_{=1}$}
\vspace{0.02cm}

Let $E_2=(C,N,\ell,A) \in \calE_{=1}$
with $\ell = 6$ rounds,
a set of voters $N=[12]$, 
a set of candidates $C = \{x, y, z, c_1,\dots,c_{12}\}$,
and the ballots of voters as presented in \Cref{tab:prop:independence:ii}.

\begin{table}[ht]
    \setlength{\tabcolsep}{3pt}
    \renewcommand{\arraystretch}{1.2}
    \small
    \begin{center}
        \begin{tabular}{c | cccc cccc cccc }
        \toprule
        & \multicolumn{12}{c}{Voter} \\
        Rounds & 1 & 2 & 3 & 4 & 5 & 6 & 7 & 8 & 9 & 10 & 11 & 12 \\
        \midrule
        1 & $x$ & $x$ & $x$ & $x$ & $y$ & $y$ & $y$ & $y$ & $z$ & $z$ & $z$ & $z$ \\
        2 & $c_1$ & $c_2$ & $c_3$ & $c_4$ & $c_5$ & $c_6$ & $c_7$ & $c_8$ & $c_9$ & $c_{10}$ & $c_{11}$ & $c_{12}$\\
        3 & $c_1$ & $c_2$ & $c_3$ & $c_4$ & $c_5$ & $c_6$ & $c_7$ & $c_8$ & $c_9$ & $c_{10}$ & $c_{11}$ & $c_{12}$\\
        4 & $c_1$ & $c_2$ & $c_3$ & $c_4$ & $c_5$ & $c_6$ & $c_7$ & $c_8$ & $c_9$ & $c_{10}$ & $c_{11}$ & $c_{12}$\\
        5 & $c_1$ & $c_2$ & $c_3$ & $c_4$ & $c_5$ & $c_6$ & $c_7$ & $c_8$ & $c_9$ & $c_{10}$ & $c_{11}$ & $c_{12}$\\
        6 & $c_1$ & $c_2$ & $c_3$ & $c_4$ & $c_5$ & $c_6$ & $c_7$ & $c_8$ & $c_9$ & $c_{10}$ & $c_{11}$ & $c_{12}$\\
        \bottomrule
        \end{tabular}
    \end{center}
    \caption{The ballots of voters in the proof of \Cref{prop:independence}, Claim ii.}
    \label{tab:prop:independence:ii}
\end{table}

Observe that this is the same election
as the one we have considered in the proof of \Cref{prop:sjr:not-satisfiable}.
Thus, we already know
that there is no sJR outcome in this election.

On the other hand, we will prove that every outcome $\outcome$ provides EJR+ and Core,
as both of these axioms do not give any satisfaction guarantee in this election
to any subset of voters.

For EJR+,
let us consider for what values of $\sigma$ and $\tau$
we can have $(\sigma, \tau)$-cohesive subsets of voters $S \subseteq N$.
Clearly, $\sigma$ can be at most $4$
as in no round more than $4$ voters agree.
However, if $\sigma > 1$,
then $\tau$ has to be equal to $1$,
because there is only one round
in which at least two voters agree.
However, even for the maximum $\sigma=4$,
a subset of voters that is $(4,1)$-cohesive
receives a satisfaction guarantee of 
$$\floor{\tau \cdot \sigma / n} = \floor{1 \cdot 4 / 12} = \floor{1/3} = 0.$$
On the other hand, if $\sigma=1$,
then even for the maximum $\tau = \ell = 6$,
we again get a satisfaction guarantee of
$$\floor{\tau \cdot \sigma / n} = \floor{6 \cdot 1 / 12} = \floor{1/2} = 0.$$
Thus, indeed, EJR+ does not give any satisfaction guarantee in election $E_2$,
which means that every outcome $\outcome$ provides EJR+.

For Core,
in order for it to give a positive satisfaction guarantee,
there has to be a subset of voters $S \subseteq N$
and a subset of rounds $T \subseteq [\ell]$
such that for every $R \subseteq T$
with $|R| = \floor{|T|\cdot |S|/n}$
there exists $\outcome'_R \in C^R$
for which the minimum satisfaction in $S$ is positive, i.e.,
$\min_{i \in S} \sat_i(\outcome'_R) > 0$.
Otherwise, the voter yielding the minimum satisfaction equal to $0$
would not benefit from the deviation to $\outcome'_R$
from any other outcome $\outcome$.
Since $|R| \le |T|\cdot |S|/n \le \ell \cdot |S|/n \le |S|/2$
and each candidate in rounds $2$--$6$
gives a positive satisfaction to at most one voter,
it is not possible to guarantee all voters in $S$ a positive satisfaction from $\outcome'_R$
if $R$ does not contain round $1$.
However, as we choose the worst $\floor{|T|\cdot |S|/n}$ rounds from $T$ to $R$,
in order to have round $1$ in $R$
we would need to have $|S|/n = 1$, i.e., $S = N$.
But if $S = N$, then
$\outcome'_R$ would have to give a positive satisfaction to $12$ voters.
This is not possible as every suboutcome
can give a summed satisfaction to all voters of at most $9$
(satisfaction of $4$ can be given in the first round
and $1$ in each consecutive round).
Thus, indeed Core does not give a positive satisfaction guarantee to any subset of voters,
which means that it is satisfied by every outcome $\outcome$.

\vspace{0.2cm}
\noindent
\textbf{Claim iii: sJR does not imply wPJR even for $\calE_{=1}$}
\vspace{0.02cm}

\noindent
Let $E_3=(C,N,\ell,A) \in \calE_{=1}$ be a temporal election with
2 candidates $C = \{a,b\}$,
2 voters $N = [2]$,
$\ell = 4$ rounds,
and the ballots of voters as presented in \Cref{tab:prop:independence:iii}.

\begin{table}[htb]
    \setlength{\tabcolsep}{3pt}
    \renewcommand{\arraystretch}{1.2}
    \small
    \begin{center}
        \begin{tabular}{c | cc }
        \toprule
        & \multicolumn{2}{c}{Voter} \\
        Rounds & 1 & 2 \\
        \midrule
        1 & $\underline{a}$ & $b$ \\
        2 & $a$ & $\underline{b}$ \\
        3 & $a$ & $\underline{b}$ \\
        4 & $a$ & $\underline{b}$ \\
        \bottomrule
        \end{tabular}
    \end{center}
    \caption{The ballots of voters in the proof of \Cref{prop:independence}, Claim iii.}
    \label{tab:prop:independence:iii}
\end{table}

Consider outcome $\outcome_3 = (a, b, b, b)$.
Clearly, $\outcome_3$ provides sJR as $\sat_i(\outcome_3) \geq 1$ for every voter $i \in N$ (thus, there cannot exist a subset of voters $S$ who collectively receive a satisfaction of less than $\min(1, \floor{\ell \cdot |S|/n})$).

It remains to show that $\outcome_3$ violates wPJR. To see this, consider the singleton subset $S = \{1\}$.
Clearly, voter $1$ agrees with itself on candidate $a$ in every round of $E_3$.
However, we find that
$$
\floor{\ell \cdot |S_1|/n} = \floor{4 \cdot 2/4} = 2 < 1 = \sat_{S}(\outcome_3).
$$ 

Therefore, $S$ is a witness of a violation of wPJR,
which means that $\outcome_3$ violates wPJR.

\vspace{0.2cm}
\noindent
\textbf{Claim iv: sEJR+ does not imply wFPJR even for $\calE_{=1}$}
\vspace{0.02cm}

\noindent
Let $E_4=(C,N,\ell,A) \in \calE_{=1}$ be a temporal election with
3 candidates $C = \{x,y,z\}$,
7 voters $N = [7]$,
$\ell =7$ rounds,
and the ballots of voters as presented in \Cref{tab:prop:independence:iv}.

\begin{table}[ht]
    \setlength{\tabcolsep}{3pt}
    \renewcommand{\arraystretch}{1.2}
    \small
    \begin{center}
        \begin{tabular}{c | ccccccc }
        \toprule
        & \multicolumn{7}{c}{Voter} \\
        Rounds & 1 & 2 & 3 & 4 & 5 & 6 & 7\\
        \midrule
        1 & \underline{$x$} & \underline{$x$} & \underline{$x$} & \underline{$x$} & \underline{$x$} & \underline{$x$} & \underline{$x$} \\
        2 & $y$ & \underline{$x$} & \underline{$x$} & \underline{$x$} & \underline{$x$} & \underline{$x$} & \underline{$x$} \\
        3 & $y$ & \underline{$x$} & \underline{$x$} & \underline{$x$} & \underline{$x$} & \underline{$x$} & \underline{$x$} \\
        4 & $y$ & \underline{$x$} & \underline{$x$} & \underline{$x$} & \underline{$x$} & \underline{$x$} & \underline{$x$} \\
        5 & \underline{$y$} & $z$ & $z$ & $x$ & $x$ & $x$ & $x$ \\
        6 & \underline{$y$} & $x$ & $x$ & $z$ & $z$ & $x$ & $x$ \\
        7 & \underline{$y$} & $x$ & $x$ & $x$ & $x$ & $z$ & $z$ \\
        \bottomrule
        \end{tabular}
    \end{center}
    \caption{The ballots of voters in the proof of \Cref{prop:independence}, Claim iv.}
    \label{tab:prop:independence:iv}
\end{table}

Consider outcome $\outcome_4 = (x, x, x, x, y, y, y)$.
First, let us show that $\outcome_4$ provides sEJR+.
Observe that since the number of rounds
is equal to the number of voters,
sEJR+ requires that
for every round $r \in [\ell]$
and a subset of voters $S$
such that $\bigcap_{i \in S} a_{i,r} \neq \varnothing$, either $o_r \in \bigcap_{i \in S} a_{i,r}$ or
there exists a voter in $S$
with the satisfaction greater or equal to $|S|$.
Observe that $\bigcap_{i \in S} a_{i,r} \neq \varnothing$
implies that there has to be a candidate $c \in C$ that is approved by all voters in $S$ in round $r$.
Let $S' = \{i \in N : c \in a_{i,r}\}$ be a subset of all voters that approve $c$ in $r$.
Then observe that if one of the conditions is satisfied by $S'$ it has to be also satisfied by $S$.
Thus,
we can exhaustively check for every round $r \in [\ell]$
and every unselected candidate in this round $c \in C$,
whether either of the conditions is satisfied for
the subset of all voters that approve $c$ in $r$.
It turns out to be the case,
which we summarize in
\Cref{tab:prop:independence:iv:sejr+}.

\begin{table}[ht]
    \setlength{\tabcolsep}{5pt}
    \renewcommand{\arraystretch}{1.2}
    \small
    \begin{center}
        \begin{tabular}{c | ccc }
        \toprule
        & \multicolumn{3}{c}{Candidate} \\
        Rounds & $x$ & $y$ & $z$ \\
        \midrule
        1 & selected &
        not approved by any voter &
        not approved by any voter \\
        2 & selected &
        $\sat_1(\outcome_4) \ge 1$ &
        not approved by any voter \\
        3 & selected &
        $\sat_1(\outcome_4) \ge 1$ &
        not approved by any voter \\
        4 & selected &
        $\sat_1(\outcome_4) \ge 1$ &
        not approved by any voter \\
        5 & $\sat_4(\outcome_4) \ge 4$ &
        selected &
        $\sat_2(\outcome_4) \ge 2$ \\
        6 & $\sat_2(\outcome_4) \ge 4$ &
        selected &
        $\sat_4(\outcome_4) \ge 2$ \\
        7 &  $\sat_2(\outcome_4) \ge 4$ &
        selected &
        $\sat_6(\outcome_4) \ge 2$ \\
        \bottomrule
        \end{tabular}
    \end{center}
    \caption{sEJR+ conditions satisfaction for each round $r \in [\ell]$ and each candidate in $C$.}
    \label{tab:prop:independence:iv:sejr+}
\end{table}

On the other hand, we can prove that wFPJR is violated.
Consider a subset of agents $S = \{2,3,4,5,6,7\}$.
To determine the satisfaction guarantee for this subset
according to wFPJR
we take the worst possible subset of rounds $R$
of size $|R|=\floor{\ell \cdot |S|/n} = \floor{7\cdot 6/7} = 6$
for which we fix a suboutcome $\outcome'_R$
so to maximize the minimum satisfaction of voters from $S$.
We can easily check that by selecting candidate $x$ in each round,
we guarantee the minimum satisfaction of $5$ to every voter in $S$,
no matter which $6$ rounds we choose to $R$.
This means that according to wFPJR,
the satisfaction of subset $S$
for the outcome should be at least $5$.
However, $\sat_S(\outcome_4) = 4$,
as there are only 4 rounds in which the selected candidate
is supported by at least one voter from $S$.
Hence, wFPJR is indeed violated.

\vspace{0.2cm}
\noindent
\textbf{Claim v: sFJR does not imply wCore even for $\calE_{=1}$}
\vspace{0.02cm}

\noindent
Let $E_5=(C,N,\ell,A) \in \calE_{=1}$ be a temporal election with
3 candidates $C = \{a,b,c\}$,
8 voters $N = [8]$,
$\ell =4$ rounds,
and the ballots of voters as presented in \Cref{tab:prop:independence:v}.

\begin{table}[ht]
    \setlength{\tabcolsep}{3pt}
    \renewcommand{\arraystretch}{1.2}
    \small
    \begin{center}
        \begin{tabular}{c | cccc cccc }
        \toprule
        & \multicolumn{8}{c}{Voter} \\
        Rounds & 1 & 2 & 3 & 4 & 5 & 6 & 7 & 8 \\
        \midrule
        1 & \underline{$a$} & \underline{$a$} & \underline{$a$} & \underline{$a$} & $b$ & \underline{$a$} & $c$ & $c$ \\
        2 & \underline{$a$} & \underline{$a$} & \underline{$a$} & \underline{$a$} & \underline{$a$} & $b$ & $c$ & $c$ \\
        3 & $a$ & $a$ & $a$ & $a$ & $a$ & $a$ & \underline{$c$} &  \underline{$c$} \\
        4 & $a$ & $a$ & $a$ & $a$ & $a$ & $a$ & \underline{$c$} &  \underline{$c$} \\
        \bottomrule
        \end{tabular}
    \end{center}
    \caption{The ballots of voters in the proof of \Cref{prop:independence}, Claim v.}
    \label{tab:prop:independence:v}
\end{table}

Consider outcome $\outcome_5 = (a, a, c, c)$.
Let us first show that it satisfies sFJR.
The satisfaction guarantee given by sFJR
to a subset $S$
is the minimum satisfaction of any voter in $S$
under the best possible suboutcome $\outcome'_R$
for any subset of rounds $R \subseteq [\ell]$
with $|R| = \floor{\ell \cdot |S|/n} = |S|/2$.
Observe that under $\outcome_5$,
there are only two voters with satisfaction 1
and the rest of voters have satisfaction 2.
Thus, if sFJR was to be violated,
we would need to have a subset of voters $S \subseteq N$
for which sFJR gives a satisfaction guarantee of at least $3$.
This means that the size of $S$ needs to be at least $6$.

Observe that voters $7$ and $8$ do not agree with any voter from $\{1,2,3,4,5,6\}$ in any of the rounds.
Hence, we can analyze these two groups of voters separately.
However, combined with the fact
that the size of $S$ needs to be at least $6$,
this means that the only subset of voters
that we need to check is $S = \{1,2,3,4,5,6\}$.
Now, observe that there are no $3$ rounds
in which all voters from $S$ agree.
Thus, for any subset of $3$ rounds $R$,
there is no suboutcome $\outcome'_R$
that would give a satisfaction of at least $3$
to each voter in $S$.
Therefore, there is no subset $S$ witnessing the violation of sFJR,
which means that $\outcome_5$ provides sFJR.

On the other hand, to show that $\outcome_5$ fails wCore, consider arbitrary subset of $3$ rounds $R$
and a suboutcome $\outcome'_R$ on these rounds in which we select candidate $a$ each time.
Observe that $\sat_i(\outcome'_R) = 3 \ge 2 = \sat_i(\outcome_5)$ for every $i \in \{1,2,3,4\}$.
Moreover, $\sat_i(\outcome'_R) \ge 2 \ge 1 = \sat_i(\outcome_5)$ for $i \in \{5,6\}$.
Thus, subset of voters $S = \{1,2,3,4,5,6\}$
is a witness that wCore is violated.

\vspace{0.2cm}
\noindent
\textbf{Claim vi: sCore does not imply wEJR+ even for $\calE_{=1}$}
\vspace{0.02cm}

\noindent
Let $E_6=(C,N,\ell,A) \in \calE_{=1}$ be a temporal election with
4 candidates $C = \{a,b,c,d\}$,
7 voters $N = [7]$,
$\ell =7$ rounds,
and the ballots of voters as presented in \Cref{tab:prop:independence:vi}.

\begin{table}[ht]
    \setlength{\tabcolsep}{3pt}
    \renewcommand{\arraystretch}{1.2}
    \small
    \begin{center}
        \begin{tabular}{c | ccccccc }
        \toprule
        & \multicolumn{7}{c}{Voter} \\
        Rounds & 1 & 2 & 3 & 4 & 5 & 6 & 7 \\
        \midrule
        1 & $a$ & $a$ & $a$ & $a$ & $a$ & $a$ & \underline{$d$} \\
        2 & \underline{$a$} & \underline{$a$} & \underline{$a$} & \underline{$a$} & $b$ & $b$ & $d$ \\
        3 & \underline{$a$} & \underline{$a$} & \underline{$a$} & $b$ & \underline{$a$} & $b$ & $d$ \\
        4 & \underline{$a$} & \underline{$a$} & \underline{$a$} & $b$ & $b$ & \underline{$a$} & $d$ \\
        5 & $a$ & $a$ & $a$ & $a$ & \underline{$b$} & \underline{$b$} & $d$ \\
        6 & $a$ & $a$ & $a$ & \underline{$b$} & $a$ & \underline{$b$} & $d$ \\
        7 & $a$ & $a$ & $a$ & \underline{$b$} & \underline{$b$} & $a$ & $d$ \\
        \bottomrule
        \end{tabular}
    \end{center}
    \caption{The ballots of voters in the proof of \Cref{prop:independence}, Claim vi.}
    \label{tab:prop:independence:vi}
\end{table}
Consider outcome $\outcome_6 = (d,a,a,a,b,b,b)$.
First, let us show that $\outcome_6$ satisfies sCore.
Since in each round,
voter 7 does not agree with any voter from $\{1,\dots,6\}$
and $\sat_7(\outcome_6) \ge 1 = \floor{\ell \cdot |\{7\}|/n}$,
for the purpose of checking whether sCore is satisfied
we can focus on subsets $S \subseteq \{1,\dots,6\}$ only.

Observe that the satisfaction from outcome $\outcome_6$
of each voter in the set $\{1,\dots,6\}$ is equal to $3$.
Thus, if sCore was not to be satisfied,
there would have to be a subset $S \subseteq \{1,\dots,6\}$,
subset of rounds $R \subseteq [\ell]$ of size
$|R| = \floor{\ell \cdot |S|/n} = |S|$
and suboutcome $\outcome'_R$ such that
$\sat_i(\outcome'_R) \ge 4$ for every $i \in S$.
Since $\sat_i(\outcome'_R) \le |R|$ for every $i \in N$, $R \subseteq [\ell]$ and $\outcome'_R \in C^R$,
we know that $S$ must be of size $4$, $5$, or $6$.
Let us consider each case separately.

Case 1, $|S|=4$.
When $S$ consists of $4$ voters,
in order for some $\outcome'_R$ to give satisfaction $4$ to all voters in $S$,
it would have to hold that voters in $S$ all agree in some $4$ rounds.
There are no $4$ voters that all agree in some $4$ rounds,
hence no set $S$ of size $4$ is a witness of sCore violation.

Case 2, $|S|=5$.
Observe that election $E_6$ is symmetric with 
regard to voters $1,2,3$ as well as $4,5,6$.
Hence, without loss of generality,
we can focus on two sets of size 5, i.e.,
$S_1 = \{1,2,3,4,5\}$ and
$S_2 = \{2,3,4,5,6\}$.
In the first round, selecting candidate $a$ to suboutcome $\outcome'_R$ would give us satisfaction 1 to each voter in $S$.
Hence, every subset $R$
that does not include round $1$
and every suboutcome $\outcome'_R$,
would clearly be dominated by suboutcome we obtain
by exchanging any round in $R$ for round $1$
and selecting $a$ there.
Thus, let us assume that the first round is in $R$
and candidate $a$ is selected in that round.

In the remaining rounds, we have to choose between selecting candidate $a$ or $b$ (selecting $c$ would never be beneficial).
Let say that we selected candidate $a$ in $t_a$ rounds and candidate $b$ in $t_b$ rounds.
Observe that $|R| = 1 + t_a + t_b$.

Then, for set $S_1 = \{1,2,3,4,5\}$,
the summed satisfaction of voters $1$, $2$, and $3$ would be equal to
\[
\sat_1(\outcome'_R) + \sat_2(\outcome'_R) + \sat_3(\outcome'_R) = 3 + 3 \cdot t_a.
\]
Since satisfaction of each voter must be at least $4$, we get that $t_a \ge 3$.
On the other hand, the total satisfaction of the remaining voters is bounded by
\[
\sat_4(\outcome'_R) + \sat_5(\outcome'_R) \le 2 + t_a + 2 \cdot t_b = |R| + 1 + t_b = 6 + t_b.
\]
In order to give a satisfaction guarantee of at least $4$ to each of the voters, we need $t_b \ge 2$.
However, if round $1$ is in $R$ and $|R| = 5$,
it is impossible to have both $t_a \ge 3$
and $t_b \ge 2$.
Thus, $S_1$ (and any symmetric subset) is not a witness of sCore violation.

For set $S_2 = \{2,3,4,5,6\}$,
the total satisfaction of voters $2$ and $3$ under $\outcome'_R$ can be expressed as
\[
\sat_2(\outcome'_R) + \sat_3(\outcome'_R) = 2 + 2 \cdot t_a.
\]
Thus, again we need $t_a \ge 3$ in order to guarantee a satisfaction of at least $4$.
For voters $4$, $5$, and $6$, we have 
\[
\sat_4(\outcome'_R) + \sat_5(\outcome'_R) + \sat_6(\outcome'_R) \le 3 + t_a + 2 \cdot t_b = |R| + 2 + t_b = 7 + t_b.
\]
This time we would need $t_b \ge 5$ in order to get a satisfaction guarantee of at least $4$,
which is clearly not possible if $t_a \ge 3$
and $|R|=5$.
Thus, $S_2$ (and any symmetric subset) is not a witness of sCore violation, which concludes the case of $|S|=5$.

Case 3, $|S|=6$.
Finally consider subset $S = \{1,2,3,4,5,6\}$.
Here, again we can assume that round 1 is in $R$
and candidate $a$ is selected to $\outcome'_R$ in that round.
Let $t_a$ and $t_b$ be the respective numbers of times
$a$ and $b$ is selected in the remaining rounds.
Then, the summed satisfaction of voters $1$, $2$, and $3$
is given by
\[
\sat_1(\outcome'_R) + \sat_2(\outcome'_R) + \sat_3(\outcome'_R) = 3 + 3 \cdot t_a.
\]
To guarantee a satisfaction of at least $4$ to all voters,
we need that $t_a \ge 3$.
On the other hand,
the summed satisfaction of voters $4$, $5$, and $6$ can be bounded by
\[
\sat_4(\outcome'_R) + \sat_5(\outcome'_R) + \sat_6(\outcome'_R) \le 3 + t_a + 2 \cdot t_b = |R| + 2 + t_b = 8 + t_b.
\]
A satisfaction of at least $4$ for all voters would imply that $t_b \ge 4$.
However, $t_a \ge 3$ and $t_b \ge 4$ cannot be both met
when $|R|=6$, thus $S$ is not a witness that sCore is violated.

Since we have exhaustively shown that there is no witness to sCore violation,
this means that $\outcome_6$ indeed provides sCore.

Finally, let us show that wEJR+ is violated.
To this end, take $S=\{1,2,3,4,5,6\}$
and observe that it is $(4,\ell)$-cohesive.
Indeed, in every round,
candidate $a$ is approved by at least $4$ voters from $S$.
Then, wEJR+ requires that either at least one voter in $S$
has a satisfaction of at least
$\floor{\ell \cdot 4/n} = \floor{7 \cdot 4/7} = 4$
or candidate $a$ is selected in the first round.
None of these conditions are satisfied by $\outcome_6$,
which means that wEJR+ is indeed violated.

\vspace{0.2cm}
\noindent
\textbf{Claim vii: sFPJR does not imply wFJR even for $\calE_{=1}$}
\vspace{0.02cm}

\noindent
Let $E_7=(C,N,\ell,A) \in \calE_{=1}$ be a temporal election with
2 candidates $C = \{a,b\}$,
3 voters $N = [3]$,
$\ell =3$ rounds,
and the ballots of voters as presented in \Cref{tab:prop:independence:vii}.

\begin{table}[bht]
    \setlength{\tabcolsep}{3pt}
    \renewcommand{\arraystretch}{1.2}
    \small
    \begin{center}
        \begin{tabular}{c | ccc }
        \toprule
        & \multicolumn{3}{c}{Voter} \\
        Rounds & 1 & 2 & 3 \\
        \midrule
        1 & $\underline{b}$ & $a$ & $a$ \\
        2 & $a$ & $\underline{b}$ & $a$ \\
        3 & $a$ & $a$ & $\underline{b}$ \\
        \bottomrule
        \end{tabular}
    \end{center}
    \caption{The ballots of voters in the proof of \Cref{prop:independence}, Claim vii.}
    \label{tab:prop:independence:vii}
\end{table}

Consider outcome $\outcome_7 = (b, b, b)$. First, let us show that $\outcome_7$ provides sFPJR. We notice that each voter has a satisfaction of 1 from $\outcome_7$ and voter $i \in N$ approves the winning candidate in round $i$. Therefore, for any subset of voters $S$ we find that $\sat_S(\outcome_7) = |S|$. Now, let $\outcome'$ be an outcome and $R \subseteq [\ell]$ be a subset of rounds such that $\outcome'_R$ witnesses a violation of sFPJR. This means that there must exist a subset of voters $S$ such that $\sat_S(\outcome_7) < \sat_i(\outcome'_R)$ for every voter $i \in S$. Given that $n = \ell$, we find that $\floor{\ell \cdot |S|/n} = |S|$, but this gives rise to a contradiction because for every voter $i \in S$,
$$
\sat_i(\outcome'_R) \leq |R| = \floor{\ell \cdot |S|/n} = |S| = \sat_S(\outcome_7).
$$
Therefore, no outcome $\outcome'$ and subset of rounds $R$ exists that can witness a violation of sFPJR, so $\outcome_7$ satisfies sFPJR.

It remains to show that $\outcome_7$ violates wFJR. To see this, consider the voter subset $S = N$. Given that $|S| = n$ it follows that $\floor{\ell \cdot |S|/n} = \ell = 3$, so for any subset of rounds $R$ of an outcome that could witness a violation of wFJR, $|R| = \floor{\ell \cdot |S|/n} = 3$. As the only subset of 3 rounds of $[\ell]$ is $[\ell]$ itself, it suffices to show that there exists an outcome $\outcome'$ wherein $\min_{j \in S} \sat_j(\outcome') > 1$ since $\sat_i(\outcome_7) = 1$ for every voter $i \in N$. However, $\outcome' = (a, a, a)$ is exactly this required outcome given that $\sat_i(\outcome') = 2$ for every voter $i \in N$, which completes the proof.

\vspace{0.2cm}
\noindent
\textbf{Claim viii: sFPJR does not imply wEJR+ even for $\calE_{=1}$}
\vspace{0.02cm}

\noindent
Let $E_8=(C,N,\ell,A) \in \calE_{=1}$ be a temporal election with
3 candidates $C = \{a,b,c\}$,
4 voters $N = [4]$,
$\ell =4$ rounds,
and the ballots of voters as presented in \Cref{tab:prop:independence:viii}.

\begin{table}[bht]
    \setlength{\tabcolsep}{3pt}
    \renewcommand{\arraystretch}{1.2}
    \small
    \begin{center}
        \begin{tabular}{c | cccc }
        \toprule
        & \multicolumn{4}{c}{Voter} \\
        Rounds & 1 & 2 & 3 & 4 \\
        \midrule
        1 & $a$ & $a$ & $a$ & \underline{$c$} \\
        2 & $\underline{b}$ & $a$ & $a$ & $c$ \\
        3 & $a$ & $\underline{b}$ & $a$ & $c$ \\
        4 & $a$ & $a$ & $\underline{b}$ & $c$ \\
        \bottomrule
        \end{tabular}
    \end{center}
    \caption{The ballots of voters in the proof of \Cref{prop:independence}, Claim viii.}
    \label{tab:prop:independence:viii}
\end{table}

Consider outcome $\outcome_8 = (c, b, b, b)$.
First let us show that it satisfies sFPJR.
According to sFPJR,
for every subset of voters $S \subseteq N$
and every subset of rounds $R \subseteq [\ell]$
with $|R| = \floor{\ell \cdot |S|/n}$,
and every suboutcome $\outcome'_R$,
subset $S$ should receive a satisfaction of at least $\min_{i \in S}\sat_i(\outcome'_R)$.
Observe that for every voter $i \in N$,
suboutcome $\outcome'_R$ can give a satisfaction to $i$ of at most $|R|$.
Hence, $\min_{i \in S}\sat_i(\outcome'_R) \le |R|$.
On the other hand,
since in outcome $\outcome_8$
in every round a different voter approves the selected candidate,
for every subset of voters $S \subseteq N$, we have that
$\sat_S(\outcome_8) = |S|$.
Combining both inequalities we get
that for every subset $S \subseteq N$,
it holds that
\[
    \min_{i \in S}\sat_i(\outcome'_R) \le |R| = \floor{\ell \cdot |S| /n} = |S| = \sat_S(\outcome_8).
\]
Thus, $\outcome_8$ indeed satisfies sFPJR.

Now, let us show that $\outcome_8$ fails wEJR+.
To this end, take subset of voters $S = \{1,2,3\}$ and consider round 1.
Observe that $\bigcap_{i \in S}a_{i,1} = \{a\} \neq \varnothing$.
Moreover, $S$ is $(2,\ell)$-cohesive as in every round at least 2 voters from $S$ approve candidate $a$.
Thus, wEJR+ requires that either a candidate from $\bigcap_{i \in S}a_{i,1}$ is selected in the first round
or there is a voter $i \in S$
with satisfaction of at least
$\floor{\ell \cdot 2 / n} = \floor{4 \cdot 2 /4} = 2$.
However, $\sat_i(\outcome_8)=1$ for every $i \in S$ and $a \neq o_1$,
so neither requirement is satisfied.
Thus, wEJR+ is violated.

\vspace{0.2cm}
\noindent
\textbf{Claim ix: sFPJR does not imply EJR even for $\calE_{=1}$}
\vspace{0.02cm}

\noindent
Let $E_9=(C,N,\ell,A) \in \calE_{=1}$ be a temporal election with
2 candidates $C = \{a,b\}$,
8 voters $N = [8]$,
$\ell =8$ rounds,
and the ballots of voters as presented in \Cref{tab:prop:independence:ix}.

\begin{table}[ht]
    \setlength{\tabcolsep}{3pt}
    \renewcommand{\arraystretch}{1.2}
    \small
    \begin{center}
        \begin{tabular}{c | cccc cccc }
        \toprule
        & \multicolumn{5}{c}{Voter} \\
        Rounds & 1 & 2 & 3 & 4 & 5 & 6 & 7 & 8 \\
        \midrule
        1 & $\underline{b}$ & $a$ & $a$ & $a$ & $\underline{b}$ & $\underline{b}$ & $\underline{b}$ & $\underline{b}$ \\
        2 & $a$ & $\underline{b}$ & $a$ & $a$ & $\underline{b}$ & $\underline{b}$ & $\underline{b}$ & $\underline{b}$ \\
        3 & $a$ & $a$ & $\underline{b}$ & $a$ & $\underline{b}$ & $\underline{b}$ & $\underline{b}$ & $\underline{b}$ \\
        4 & $a$ & $a$ & $a$ & $\underline{b}$ & $\underline{b}$ & $\underline{b}$ & $\underline{b}$ & $\underline{b}$ \\
        5 & $a$ & $a$ & $a$ & $a$ & $\underline{b}$ & $\underline{b}$ & $\underline{b}$ & $\underline{b}$ \\
        6 & $a$ & $a$ & $a$ & $a$ & $\underline{b}$ & $\underline{b}$ & $\underline{b}$ & $\underline{b}$ \\
        7 & $a$ & $a$ & $a$ & $a$ & $\underline{b}$ & $\underline{b}$ & $\underline{b}$ & $\underline{b}$ \\
        8 & $a$ & $a$ & $a$ & $a$ & $\underline{b}$ & $\underline{b}$ & $\underline{b}$ & $\underline{b}$ \\
        \bottomrule
        \end{tabular}
    \end{center}
    \caption{The ballots of voters in the proof of \Cref{prop:independence}, Claim ix.}
    \label{tab:prop:independence:ix}
\end{table}

Consider outcome $\outcome_9 = (b, b, b, b, b, b, b, b)$. First, let us show that $\outcome_9$ provides sFPJR. Given that voters 5, 6, 7 and 8 approve the winner in every round of $\outcome_9$, none of these voters can be more satisfied by another outcome, so any subset of voters containing one of these four voters cannot witness a violation of sFPJR. It therefore remains to check that no subset of voters $S \subseteq \{1, 2, 3, 4\}$ can violate sFPJR.

Observe that for each voter $i \in \{1, 2, 3, 4\}$, $\sat_i(\outcome_9) = 1$ and for any two voters in the subset $\{1, 2, 3, 4\}$, there exists no round in outcome $\outcome_9$ wherein both voters are satisfied by the winning candidate. Therefore, for any subset of voters $S \subseteq \{1, 2, 3, 4\}$, we find that $\sat_S(\outcome_9) = |S|$.

Now we assume that there exists an outcome $\outcome'$ that witnesses a violation of sFPJR. We would therefore have that for some subset of voters $S \subseteq \{1, 2, 3, 4\}$, $\sat_S(\outcome_9) < \sat_i(\outcome'_R)$ for each $i \in S$ where $|R| = \floor{\ell \cdot |S|/n} = |S|$ given that $n = \ell$. However, this is impossible because $\sat_i(\outcome'_R) \leq |R| = |S| = \sat_S(\outcome_9)$ for every voter $i \in S$.
Therefore, no outcome can bear witness to a violation of sFPJR and we deduce that $\outcome_9$ satisfies sFPJR.

To show that $\outcome_9$ violates EJR, we consider the subset of voters $S = \{1, 2, 3, 4\}$, which agree on candidate $b$ in rounds 5 to 8, hence $t = 4$.
Thus, EJR requires that at least one voter from $S$ should receive a satisfaction of at least $\floor{t \cdot |S|/n} = \floor{4 \cdot 4/8} = 2$. However, for every voter $i \in S$ we find that $\sat_i(\outcome_9) = 1$, thereby violating EJR.

\vspace{0.2cm}
\noindent
\textbf{Claim x: sFPJR does not imply wEJR for $\calE_{\geq 1}$}
\vspace{0.02cm}

\noindent
Let $E_0=(C,N,\ell,A) \in \calE_{\geq 1}$ be a temporal election with
5 candidates $C = \{a,b,c,d,e\}$,
6 voters $N = [6]$,
$\ell =6$ rounds,
and the ballots of voters as presented in \Cref{tab:prop:independence:x}.

\begin{table}[ht]
    \setlength{\tabcolsep}{3pt}
    \renewcommand{\arraystretch}{1.2}
    \small
    \begin{center}
        \begin{tabular}{c | ccc ccc }
        \toprule
        & \multicolumn{6}{c}{Voter} \\
        Rounds & 1 & 2 & 3 & 4 & 5 & 6 \\
        \midrule
        1 & $\{a, \underline{b}, c\}$ & $\{a, \underline{b}, d\}$ & $\{a, c, d\}$ & $\{\underline{b}, c, d, e\}$ & $\{\underline{b}, c, d, e\}$ & $\{\underline{b}, c, d, e\}$ \\
        2 & $\{a, b, \underline{c}\}$ & $\{a, b, d\}$ & $\{a, \underline{c}, d\}$ & $\{b, \underline{c}, d, e\}$ & $\{b, \underline{c}, d, e\}$ & $\{b, \underline{c}, d, e\}$ \\
        3 & $\{a, b, c\}$ & $\{a, b, \underline{d}\}$ & $\{a, c, \underline{d}\}$ & $\{b, c, \underline{d}, e\}$ & $\{b, c, \underline{d}, e\}$ & $\{b, c, \underline{d}, e\}$ \\
        4 & $\{a, b, c\}$ & $\{a, b, d\}$ & $\{a, c, d\}$ & $\{b, c, d, \underline{e}\}$ & $\{b, c, d, \underline{e}\}$ & $\{b, c, d, \underline{e}\}$ \\
        5 & $\{a, b, c\}$ & $\{a, b, d\}$ & $\{a, c, d\}$ & $\{b, c, d, \underline{e}\}$ & $\{b, c, d, \underline{e}\}$ & $\{b, c, d, \underline{e}\}$ \\
        6 & $\{a, b, c\}$ & $\{a, b, d\}$ & $\{a, c, d\}$ & $\{b, c, d, \underline{e}\}$ & $\{b, c, d, \underline{e}\}$ & $\{b, c, d, \underline{e}\}$ \\
        \bottomrule
        \end{tabular}
    \end{center}
    \caption{The ballots of voters in the proof of \Cref{prop:independence}, Claim x.}
    \label{tab:prop:independence:x}
\end{table}

Consider outcome $\outcome_0 = (b, c, d, e, e, e)$. First, let us show that $\outcome_0$ provides sFPJR. Given that voters 4, 5 and 6 approve the winner in every round of $\outcome_0$, none of these voters can be more satisfied by another outcome, so any subset of voters containing one of these three voters cannot witness a violation of sFPJR. It therefore remains to check that no subset of voters $S \subseteq \{1, 2, 3\}$ can violate sFPJR.

Given that $\sat_i(\outcome_0) = 2$ for every voter $i \in \{1, 2, 3\}$, any outcome $\outcome'$ that could violate sFPJR would have to provide a satisfaction of at least 3 to every voter. This would require us to find a suboutcome that contains at least 3 rounds (i.e.\ $\outcome'_R$ where $|R| \geq 3$), but as $n = \ell$ this ensures that $|S| \geq 3$ since $|R| = \floor{\ell \cdot |S|/n} = |S|$. Therefore, the only subset of voters that we need to check is $S = \{1,2,3\}$. However, given that $\sat_{S}(\outcome_0) = 3$ and any suboutcome of 3 rounds can provide a satisfaction of at most 3 to any voter, $S$ is not a witness of sFPJR violation. Since there is no such witness, $\outcome_0$ provides sFPJR.

It remains to show that $\outcome_0$ violates wEJR. To see this, we again consider the subset of voters $S = \{1, 2, 3\}$. Every voter in $S$ agrees on candidate $a$ in all 6 rounds, hence in order for $\outcome_0$ to satisfy wEJR, there must exist a voter $i \in S$ such that
$$\sat_i(\outcome_0) \geq \floor{\ell \cdot |S|/n} = \floor{6 \cdot 3/6} = 3.$$
However, each voter $i \in S_1$ only receives a satisfaction of 2 from $\outcome_0$, hence $\outcome_0$ violates wEJR. 
\end{proof}
\end{document}